\documentclass[11pt]{article}
\usepackage{amsmath, amssymb, amsfonts, amsthm}
\usepackage{graphicx}

\newtheorem{theorem}{Theorem}[section]

\newtheorem{lemma}[theorem]{Lemma}

\newcommand{\rd}{{\rm d}}
\newcommand{\be}{\begin{equation}}
\newcommand{\ee}{\end{equation}}
\newcommand{\bey}{\begin{eqnarray}}
\newcommand{\eey}{\end{eqnarray}}

\newcommand{\eps}{\varepsilon}

\newcommand{\cU}{{\cal U}}

\newcommand{\bR}{{\mathbb R}}

\newcommand{\bN}{{\mathbb N}}

\newcommand{\wt}{\widetilde}
\newcommand{\wh}{\widehat}

\newcommand{\const}{\mathrm{const}}

\newcommand{\cE}{{\cal E}}

\newcommand{\cL}{{\cal L}}

\newcommand{\cO}{{\cal O}}

\newcommand{\supp}{\operatorname{supp}}

\input epsf

\newcommand{\donothing}[1]{}

\begin{document}

\title{Ionization of Atoms by Intense Laser Pulses}

\author{J\"{u}rg Fr\"{o}hlich${}^1$, Alessandro Pizzo${}^2$\thanks{Supported by NSF grant DMS-0905988.}, Benjamin Schlein${}^3$  \\ \\
\small{1. Institute~of~Theoretical Physics, ETH Z\"{u}rich} \\
\small{CH-8093 Z\"{u}rich, Switzerland} \\ \small{juerg@itp.phys.ethz.ch} \\ \\
\small{2. Department of Mathematics, University of California Davis} \\
\small{One Shields Avenue, Davis, California 95616, USA} \\
\small{pizzo@math.ucdavis.edu}
\\ \\
\small{3. Institute for Applied Mathematics, University of Bonn}\\  
\small{Endenicher Allee 60, 53115 Bonn, Germany} \\ \small{b.schlein@dpmms.cam.ac.uk}}

\date{}

\maketitle

{\it Dedicated to our friend and colleague Robert Schrader on the occasion of his $70^{th}$ birthday}

\bigskip

\begin{abstract}
The process of ionization of a hydrogen atom by a short infrared laser pulse is studied in the regime of very large pulse intensity, in the dipole approximation. Let $A$ denote the integral of the electric field of the pulse over time at the location of the atomic nucleus. It is shown that, in the limit where $|A| \to \infty$, the ionization probability approaches unity and the electron is ejected into a cone opening in the direction of $-A$ and of arbitrarily small opening angle. Asymptotics of various physical quantities in $|A|^{-1}$ is studied carefully. Our results are in qualitative agreement with experimental data reported in \cite{1,2}.
\end{abstract}

\section{Experimental Findings and Preliminary Theoretical Considerations}
\label{sec:intro}

In recent experimental work \cite{1}, \cite{2}, P. Eckle et al. have investigated the ionization of Helium atoms by highly intense elliptically polarized infrared laser pulses of short duration. One of the purposes of their work has been to perform an (indirect) measurement of the \emph{tunneling delay time} in strong-field ionization of Helium atoms. The experimental parameters in their work have been chosen as follows: the pulse duration, $T$, is around $5.5\,\, \text{femtoseconds}$; the peak intensity, $I_0$, is between $2.3\times 10^{14}$ and $3.5\times 10^{14}$ watts per square centimeter, and the center wave-length is around $725\,\,{\text{nm}}$. The ionization potential, $I_p$, of a Helium atom in its groundstate is known to be $I_p \approx 24.6\,\,{\text{eV}}$. These parameter values yield a \emph{Keldysh parameter}, $\gamma$, for circularly polarized light ranging from $1.17$ to $1.45$. The Keldysh parameter for circular polarization is given by
\begin{equation}
\gamma \approx  0.33 \sqrt{\frac{I_p(eV)}{I_0(10^{14}W/cm^2)[\cL(\mu m)]^2}}\,.
\end{equation}
If $\gamma \gg 1$, i.e., for short wave lengths, $\cL$, and low intensity, $I_0$, the ionization process can be described in terms of multi-photon absorption, and one may attempt to treat the ionization problem perturbatively; (for a theoretical analysis of a related problem, see, e.g., \cite{3}).

\medskip

If $\gamma \ll 1$, i.e., for high intensities and long wave lenghts, a regime is approached where the electromagnetic field can be treated classically. However, due to the high intensity of the pulse, the theoretical analysis of the ionization process is \emph{intrinsically non-perturbative} in the coupling of the electrons to the electromagnetic field. This is the regime we study in this paper.

\medskip

For the values of $\gamma$ between $1.17$ and $1.45$ realized in the experiments described in \cite{1}, \cite{2}, reliable analytical calculations of the ionization process appear  to be very difficult to come by, and it is advisable to perform numerical studies; see \cite{4}. We find, however, that our analytical results are in good qualitative agreement with the experimental findings in \cite{1}, \cite{2}. One key point of these findings is that the ionization process of a Helium atom by a short, intense near-infrared laser pulse is essentially instantaneous, in contrast to theoretical predictions based on an approximate theoretical picture taken from \cite{5}, \cite{6}: Experimentally, an upper bound on the time it takes to ionize a Helium atom (with experimental parameters chosen as discussed above) appears to lie between $12$ and $34\,\, \text{attoseconds}$, while a theoretical prediction relying on  \cite{5}, \cite{6} yields an ionization (or ``barrier traversal") time of $450-560\,\, \text{attoseconds}$. Obviously there is a problem with either the interpretation of the experimental findings in terms of an ``ionization time" or with the approximate theory of the ionization process based on \cite{5}, \cite{6}; but most likely with both. The purpose of our paper is to provide a qualitative theoretical interpretation of the data gathered in the experiments described in \cite{1}, \cite{2}.

\medskip

We start with a brief sketch of the picture on which the theoretical interpretation of the experimental results is based that the authors of \cite{1} have advocated implicitly. We then describe our own approach and state our main results.

\medskip

Without harm, we may simplify our discussion by considering the ionization of Hydrogen atoms or Helium$^{+}$ ions by elliptically polarized laser pulses. The direction of propagation of the pulses through a very dilute, cold gas of atoms or ions is chosen to be our $z$-axis. The electric and magnetic field of the pulse are then parallel to the $x-y$ plane. If $\cE_{0}$ denotes the peak electric field of the pulse at the location of an atom or ion and $T$ denotes the duration of the pulse then the field of the pulse is assumed to be homogeneous over a region of the $x-y$ plane of large diameter, $d$, as compared to $\cE_{0} T^2$, centered at the location of the atom or ion. Note that $\cE_{0} T^2$ has the dimension of length. This assumption partially justifies to use the dipole approximation. 

\medskip

The Hamiltonian generating the time evolution of the electron in the atom or ion then only depends on the electric field, $E(t)$, at the location of the atomic (or ionic) nucleus; ($t$ denotes time). The vector $E(t)$ can be chosen to have the form
\begin{equation}\label{eq:2}
E(t)=\cE_0 (t) \left( \cos \left[\omega\left(t-\frac{T}{2}\right)\right], \epsilon \sin\left[\omega \left(t-\frac{T}{2}\right)\right], 0\right)\,
\end{equation}
where $\cE_0 (t)$ is a smooth envelope function with support in the interval $[0,T]$, $\omega=2\pi c/\cL$ is the angular frequency of the pulse (with $\cL \ll cT$), and $\epsilon$ is a parameter describing the elliptical polarization of the pulse. To be concrete, we choose $\cE_0(t)$ to be non-negative, symmetric-decreasing about $t=T/2$, with a maximum, $\cE_0(T/2)=:\cE_0$, at $t=T/2$.

\medskip

Apparently, the pulse arrives at the location of the nucleus at time $t=0$ and lasts until time $t=T$. An important quantity is the vector potential
\begin{equation}\label{eq:3}
A(t)=\int_0^{t}d\tau \,E(\tau)\,.
\end{equation}
Clearly, $A(t)=0$, for $t\leq0$, and $A(t)\equiv A(T)$, for $t\geq T$. For our choice of the envelope function $\cE_0(t)$,
\begin{equation}\label{eq:4}
A(T)=\,\const \cdot \cE_0 \, (1,0,0)\,,
\end{equation}
where the constant depends on $\omega$ and on $\cE_0(t)$; it tends to $0$ rapidly, as $\omega\to \infty$, i.e., in the ultraviolet. In this regime, the Keldysh parameter $\gamma$ becomes very large, and the analysis presented in our paper is not applicable. It does, however, apply to the situation where $\const\, \cdot \cE_0$, in Eq. (\ref{eq:4}), becomes \emph{large}, meaning that $\gamma$ becomes \emph{small}.

\medskip

To anticipate our main result, we will show that, for a laser pulse of the form in  Eq. (\ref{eq:2}),
\begin{itemize}
\item [i)]
the ionization probability approaches unity, as $\cE_0\to \infty$ (with a rate that will be estimated explicitly), and 
\item[ii)]
the electron is ejected by the pulse into a cone with axis parallel to $A(T)$ and a  small opening angle $\Theta=\Theta(\cE_0)$; its average velocity $v=v(\cE_0)$ is approximately parallel to $A(T)$. Moreover,
\begin{equation}\label{eq:5}
\Theta(\cE_0)\to 0\,,\, \text{as}\,\,\cE_0\to \infty\,,
\end{equation}
(with a rate that will be estimated), and 
\begin{equation}\label{eq:6}
v(\cE_0)\parallel A(T) \,,\, \text{as}\,\,\cE_0\to \infty\,,
\end{equation}
with $|v(\cE_0)|\propto \cE_0$.
\end{itemize}

These theoretical results are in good qualitative agreement with the experimental findings described in \cite{1}, \cite{2}. In the experiments, the motion of the \emph{ions} after ionization is measured. However, by momentum conservation, such measurements also determine the motion of the electron.

\medskip

In \cite{1}, data compatible with Eqs. (\ref{eq:5}) and (\ref{eq:6}) are interpreted as saying that the ionization process is nearly instantaneous. This interpretation is based, implicitly, on arguments that rely on the ``Ritz Hamiltonian" for the motion of the electron:
\begin{equation}\label{eq:7}
H_{Ritz}(t)=-\Delta-\frac{Z}{|x|}-E(t)\cdot x\,.
\end{equation}
$\Delta$ is the Laplacian, $Z$ is the charge of the nucleus, and $E(t)$ is the electric field of the laser pulse at the location of the nucleus, (see Eq. (\ref{eq:2})).
Here we work in units such that $\hbar=1$, $m_{el}=1/2$ and $e=1$, where $m_{el}$ is the mass of an electron and $e$ is the elementary electric charge. Therefore, in our units, the numerical value of the speed of light, $c$, is around $137$. Hereafter, we follow the convention that the dimension of a physical quantity is a function  of  the length only, namely: $\text{[length]}=length$; $\text{[mass]}=length^{-1}$; $[\text{time}]=length$;  the electric charge is \emph{dimensionless}.

\medskip

At a fixed moment, $t=t_0$, of time, the potential
\begin{equation}
U_{t_0}(x):=-\frac{Z}{|x|}-E(t_0)\cdot x
\end{equation}
has a shape indicated in Fig. \ref{fig}.

\begin{figure}
\begin{center}
\label{fig}
\includegraphics[height=1.5in]{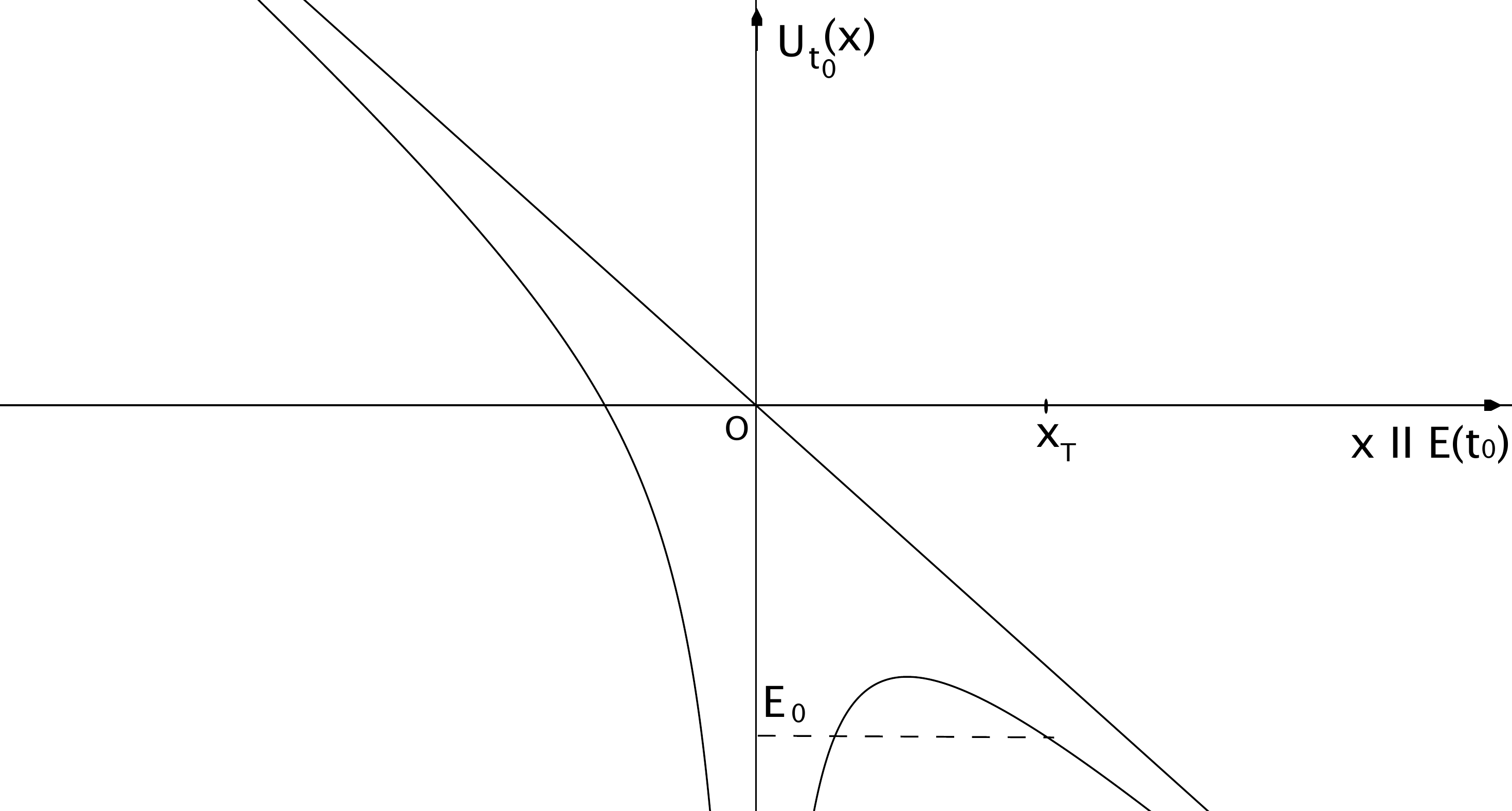}
\caption{The potential $U_{t_0} (x)$.}
\end{center}
\end{figure}

\medskip

Initially, the electron is localized near the nucleus placed at the origin, $O$, of our coordinate system and treated as static for the duration of the tunneling process. If $E(t)$ depends slowly on time $t$, i.e., for rather large pulse duration $T$ and long wave lengths, one may expect that an adiabatic approximation for the description of the tunneling process of the electron through the barrier of the potential $U_{t_0}(x)$ to the point $x_{T}$ (see Fig. \ref{fig}) is appropriate. If $\Delta t_{T}$ denotes the barrier traversal time, the electric field acting on the (nearly free) electron, after it has traversed the barrier, is given by $E(t)$, with $t\geq t_0+\Delta t_{T}$. If we interpret $t_0=0$  as the time of onset of barrier traversal then the electron, after barrier traversal, will be ejected in a direction roughly parallel to the vector
\begin{equation}
X:=\int_{t_0+\Delta t_{T}}^{T}d\tau E(\tau)\,.
\end{equation}
For a pulse described by Eq. (\ref{eq:2}) and a strictly positive barrier traversal time, $\Delta t_{T}$, the direction of $X$ in which the electron is ejected is \emph{not} parallel to the direction of $A(T)$ (parallel to the $x$-axis, for our concrete choice of an envelope function $\cE_{0}(t)$). By tuning the direction of $A(T)$ and measuring the direction in which the electrons are ejected, one can determine the angle, $\phi$, between $X$ and $A(T)$. This angle then provides information on the barrier traversal time $\Delta t_{T}$. Experimentally, $\phi$ is very small, so that $\Delta t_{T}$ is argued to be very short.

\medskip

The analysis presented in this paper shows that, for large $\cE_0$, $\phi$ is small. We have found the Ritz Hamiltonians in Eq. (\ref{eq:7}) to be rather inconvenient for an analysis of ionization processes. It is advantageous to, instead, consider the ``Kramers Hamiltonians"
\begin{equation}\label{eq:8}
H(t)=(p-A(t))^2-\frac{Z}{|x|}\,,
\end{equation}
where $p=-i\nabla$ is the usual electron momentum operator and $A(t)$ is the vector potential at the location of the nucleus given in Eq. (\ref{eq:3}). The evolutions generated by $H_{Ritz}(t)$ (see (\ref{eq:7}))  and $H(t)$, as in (\ref{eq:8}), are related to each other by a time-dependent gauge transformation given by
\begin{equation}\label{eq:9}
\Lambda (x,t):=A(t)\cdot x\,.
\end{equation} 
If $(E(T)\cdot x,0)$ denotes the $4-$vector potential before the gauge transformation (\ref{eq:9}) is made then, after this gauge transformation, it is given by $(0, A(t))$. Quantum-mechanically, the gauge equivalence of the time evolutions generated by the Ritz Hamiltonians, Eq. (\ref{eq:7}), and the Kramers Hamiltonians, Eq. (\ref{eq:8}), can easily be verified using the Trotter product formula (see, e.g., \cite{7}) for the propagators and the identity
\begin{equation}
e^{-i\Lambda(x,t)}H(t)e^{i\Lambda(x,t)}=p^2-\frac{Z}{|x|}\,,
\end{equation}
with $H(t)$ as in (\ref{eq:8}).

\medskip

Next, we sketch some key ideas in our analysis of the time evolution generated by the Kramers Hamiltonians. As an initial condition, $\psi_0$, for the electron we choose a bound state wave function, typically the atomic groundstate. In our units, it has a spatial spread of order $\mathcal{O}(Z^{-1})$. The quantum-mechanical propagator generated by the Kramers Hamiltonians $H(t)$, defined in eq. (\ref{eq:8}), is denoted by $\cU(t,t_0)=\cU(t,t_0;Z)$. It evolves an electronic wave function from time $t_0$ to time $t$ and solves the equation
\begin{equation}\label{eq:10}
i\partial_t \, \cU(t,t_0;Z)=H(t) \, \cU(t, t_0;Z)\,,
\end{equation}
with $\cU(t_0,t_0;Z)=1$, for an arbitrary $t_0$; see \cite{8}. We note that the propagator $\cU_0 (t,t_0)\equiv \cU(t,t_0;Z=0)$ can be calculated explicitly:
\begin{eqnarray}
\cU_0 (t,t_0)&=&\exp{[-i\int_{t_0}^{t}(p-A(\tau))^2d\tau]}\\
& =&e^{i\phi(t,t_0)}e^{-i(t-t_0)p^2}\exp{[2ip\cdot\int_{t_0}^{t}A(\tau)d\tau]}\,.\label{eq:11}
\end{eqnarray}
The first factor on the R. S. of (\ref{eq:11}) is a pure phase factor (with $\phi(t,t_0)=-\int_{t_0}^{t}A(\tau)^2d\tau$), the second factor is the free time evolution, and the third factor is a space translation by the vector $2\int_{t_0}^{t}A(\tau)d\tau$.

\medskip

As our initial time, we choose $t_0=0$, and the initial condition at $t=0$ is chosen to be $\psi_0$, as described above. The laser pulse hits the atom at time $t=0$ and lasts up to time $T$. Because of the space translation,
\begin{equation}\label{eq:12}
T_{\cE_0}:=\exp{[2ip\cdot\int_{0}^{T}A(\tau)d\tau]}\,,
\end{equation}
in the free propagator (\ref{eq:11}), which moves the initial wave function, $\psi_0$, far out of the potential well (described by $-Z/|x|$), provided $\cE_0$ (the peak electric field) is large, one expects that
\begin{equation}\label{eq:13}
\cU(T,0;Z)\psi_0 \approx \cU_{0}(T,0)\psi_0\,,
\end{equation}
with an error term that tends to $0$, as $\cE_0\to\infty$. Results of this type have first been proven by Fring, Kostrykin and Schrader in \cite{FKS}. We will reproduce their results in Sect. 2, below.

\medskip

As noted in (\ref{eq:3}),
\begin{equation}
A(t)=A(T)\, ,\, \text{for}\,\,t\geq T\,,
\end{equation}
i.e., the vector potential is constant when the pulse has passed. We may therefore use a gauge transformation to remove it:
\begin{equation}\label{eq:15}
e^{-i\Lambda(\,.\,,\, T)} \cU (t,T;Z)e^{i\Lambda(\,. \,,\, T)}=\cU_{C} (t,T),\, \text{for all}\,\, t\geq T\,,
\end{equation}
where $\cU_{C}(t,T)=\exp{[-i(t-T)H_{C}]}$, and
\begin{equation}\label{eq:16}
H_{C}:=p^2-\frac{Z}{|x|}
\end{equation}
is the Coulomb Hamiltonian.

\medskip

Next we note that, by Eq. (\ref{eq:9}),
\begin{equation}\label{eq:17}
e^{-i\Lambda(x\,,\, T)}=e^{-iA(T)\cdot x},\,
\end{equation}
i.e., $e^{-i\Lambda(x\,,\, T)}$ is a translation in momentum space: it translates $\widehat{\psi}_{T}(p)$ to
\begin{equation}\label{eq:18}
\widehat{\psi}_{A(T)}(p):=\widehat{\psi}_{T}(p+A(T))\,,
\end{equation}
where
\begin{equation}
\psi_{T}(x)=(\cU(T,0;Z)\psi_0)(x)\,,
\end{equation}
and $\widehat{\psi}_{T}$ is the Fourier transform of $\psi_{T}$. An electron in the state given by $\psi_{A(T)}$, see Eq. (\ref{eq:18}), has a mean distance from the nucleus of order $\mathcal{O}(|\int_{0}^{T}A(\tau)d\tau\,|)$ and a mean velocity in the direction of $A(T)$ of magnitude $|A(T)|$. Thus, the mean distance of $\psi_{A(t)}$ from the nucleus and the mean velocity of the electron, parallel to $A(T)$, diverge, as the peak electric field, $\cE_0$, of the pulse tends to $\infty$. However, by Eqs. (\ref{eq:13}) and (\ref{eq:11}), the spread of the wave function $\psi_{A(t)}$ in $x-$space around its mean position is of order $\mathcal{O}(TZ)$, which is \emph{independent} of $\cE_0$. It is then almost obvious that, for $t\geq T$, 
\begin{eqnarray}
\cU(t,0;Z)\psi_0&=&\cU (t,T;Z)\psi_T \\
&=&e^{i\Lambda(\,.\,,\,T)} \cU_{C}(t,T) \psi_{A(T)}\\
&\approx& e^{i \Lambda(\,.\,,\,T)}e^{-i(t-T)p^2}\psi_{A(T)}\,,\label{eq:19}
\end{eqnarray}
with an error term that tends to $0$, as $\cE_0\to \infty$, \emph{uniformly} in $t\geq T$. This will be proven mathematically in Sect. 2.2., below. The phase factor, $e^{i\Lambda(.,T)}$, on the R.S. of (\ref{eq:19}) is unimportant. Moreover, $\exp{[-i(t-T)p^2]}\psi_{A(T)}$ is the \emph{free} time evolution of an electron wave function initially located at a distance of order $\mathcal{O}(|\int_{0}^{T}A(\tau)d\tau\,|)$ from the nucleus and with a mean velocity parallel to $A(T)$ and of magnitude $|A(T)|$.  Its spread in the direction perpendicular to $A(T)$ is of order $\mathcal{O}(t\,Z)$, which is \emph{independent} of $\cE_0$. Thus, the state $\cU(t,0;Z)\psi_0$ propagates into a cone with axis parallel to $A(T)$ and with an opening angle of order $\mathcal{O}(Z/|A(T)|)$, which tends to $0$, as $\cE_0 \to \infty$.

\medskip

In the technical sections of this paper, these claims are verified mathematically, and the asymptotics in $1/\cE_0$ is estimated quite carefully. This is crucial, because the Kramers Hamiltonians $H(t)$ of Eq. (\ref{eq:8}) do \emph{not} capture the physics of the ionization process correctly for very large values of $\cE_0$, for the following reasons:
\begin{itemize}
\item[(1)]
Non-relativistic kinematics for the electron is justified in our study of the ionization process \emph{only} if the (mean) electron speed after ionization, $|A(T)|$, is \emph{small} compared to the speed of light, $c$, (with $c \approx 137$, in our units). If this condition is violated relativistic kinematics would have to be employed, and electron-positron pair creation by the laser pulse in the Coulomb field of the nucleus would have to be incorporated in our analysis, i.e., the whole process would have to be studied by using methods of relativistic QED.
\item[(2)]
The dipole approximation used in the Hamiltonians defined in Eqs. (\ref{eq:7}) and (\ref{eq:8}) can only be justified under the following conditions:
\begin{itemize}
\item[(i)]
The wave length $\cL$ and the spatial extension, $Tc$, of the laser pulse in the propagation direction (here the $z-$axis) must be large, as compared to the spatial spread in the $z-$direction of the electron wave function at time $t=T$, which is of order  $\cO(TZ)$. It follows right away that $Z\ll 137$, i.e., our analysis only applies to light atoms, such as Hydrogen or Helium, which, of course, was to be expected. Thus, we must impose that
\begin{equation}\label{eq:20}
TZ\ll \cL \ll Tc
\end{equation}
\item[(ii)]
In order to justify neglecting the spatial dependence of the vector potential, $A(x,t)$, of the laser pulse in the \emph{Pauli-Fierz Hamiltonian}
\begin{equation}
H_{PF}(t):=(p-A((x,t))^2-\frac{Z}{|x|}
\end{equation}
that should be used in our analysis, instead of the Kramers Hamiltonian, Eq. (\ref{eq:8}), the laser pulse must be spatially homogenous in the $x-$ and the $y-$directions up to a distance $d$ from the nucleus \emph{large} compared to the mean distance of the electron from the nucleus at time $T$, which is given by $2|\int_0^{T}A(\tau)d\tau|$.
\item[(iii)]
Finally, terms like $|A(x,t)^2-A(0,t)^2|$ should be small in the tales of the electron wave function, $\psi_t$, for all times. These conditions are satisfied if $Z\ll 137$ and if $\cE_0$ is fairly small compared to $Z$; e.g., $Z$ and $\cE_0$ of order $1$.
\end{itemize}
\end{itemize}
Since our analytical methods only yield asymptotics in $1/\cE_0$, we would be lucky if our results gave reliable information about the ionization process for $\cE_0$ of order $1$, (i.e., $\gamma \approx 1$), corresponding to the experimental situation and needed to justify the dipole approximation. More precise quantitative information can presumably only be obtained from extensive numerical simulations.

\medskip

Yet, it is gratifying to note that our results are in good qualitative agreement with the experimental findings. Moreover, our analysis, which is based on the Kramers Hamiltonian in Eq. (\ref{eq:8}), suggests that naive calculations of ``barrier traversal times" based on an adiabatic approximation to the Ritz Hamiltonians, Eq. (\ref{eq:7}), may not yield reliable results.

\bigskip

{\bf Acknowledgements.} We thank Patrissa Eckle and Ursula Keller for explaining their experiments to us and encouraging us to carry out the analysis presented in this paper.

\section{Description of the Theoretical Setup}
\setcounter{equation}{0}

We consider an electron bound to a nucleus by a static potential $V(x)$ and under the influence of a laser pulse described, in the Coulomb gauge, by the time dependent vector potential $A(t)$, which we assume to be independent of $x$. The Hamiltonian is given by
\[ H(t) = (p-A(t))^2 + V(x) \]
and acts on the Hilbert space $L^2 (\bR^3)$. Here $p= -i\nabla$ is the momentum operator. 
We denote by $\cU (t,s)$ the propagator generated by the time-dependent Hamiltonian $H(t)$, that is \begin{equation}\label{eq:ev} i\partial_t \cU (t,s) = H(t) \cU (t,s) , \qquad \text{with } \quad \cU(s,s) = 1 \quad \text{for all } s \in \bR. \end{equation}

\subsection{The Pulse}

We consider a pulse with amplitude $\lambda$ lasting for a time $T >0$. We will be interested here in fixed $T$ and large $\lambda$.  

The electric component of the pulse is given by \[ E(t) = \frac{\lambda}{T} f(t/T) \] for a vector valued function $f: \bR \to \bR^3$, with $\supp \, f \subset [0,1]$. (In Section \ref{sec:intro}, we have used the notation $\cE_0 \simeq \lambda/T$). The vector potential $A(t)$ is then given by \[ A(t) = \int_{-\infty}^t \rd s \, E(s) =  \lambda F(t/T) \]
with \[ F(s) = \int_{-\infty}^s \rd \tau f(\tau) \] By definition $F(s) = 0$, for all $s < 0$, and $F(s) = F(1)$, for all $s \geq 1$. 

The time integral of the vector potential will also play an important role in our analysis. We set \[ G(s) = \int_{-\infty}^s \rd \tau \, F(\tau) \] Then
\[ \int_{-\infty}^t A(s) \rd s = \lambda T G (t/T). \]
By definition $G(s) = 0$, for all $s <0$, and $G(s) = G(1) + (s-1) F(1)$, for all $s >1$. 

\medskip

{\it Assumptions on Pulse.}  We assume that
\begin{equation}\label{eq:ass0}
|G(s)|^{-1} \in L^1 ((s_0,1)), \qquad \text{for all } 0 < s_0 < 1 \, . 
\end{equation}
Moreover, we assume that 
\begin{equation}\label{eq:ass1}
F(1)  \not = 0 
\end{equation}
and that 
\begin{equation}\label{eq:ass2}
|G(s)| \geq C s \, , \qquad \text{for all } s \geq 1\, .
\end{equation}
Assuming that $F(1) \not = 0$, this last condition is satisfied if $F(1) \cdot G(1) \geq 0$;  in other words, if the angle between $F(1)$ and $G(1)$ is less or equal to $\pi$. In fact, for arbitrary $s \geq 1$, 
\[ \begin{split} |G(s)|^2 = & \; |G(1) + (s-1) F(1)|^2 = |G(1)|^2 + (s-1)^2 |F(1)|^2 + 2 (s-1) G(1) \cdot F(1) \\ \geq &\; |G(1)|^2 + (s-1)^2 |F(1)|^2 \geq \frac{\min (|G(1)|^2, |F(1)|^2)}{2} s^2 \, . \end{split} \]

\medskip

{\it Examples.} A simple example of a pulse satisfying the assumptions (\ref{eq:ass1}), (\ref{eq:ass2}) is obtained by setting \[ f (s) = {\bf \eps} \, {\bf 1} (0 \leq s \leq 1) \] for a fixed polarization vector ${\bf \eps} \in \bR^3$ (pulse with linear polarization). Then $F(s) = 0$, for $s \leq 0$, $F(s) = {\bf \eps}\, s$, for $s\in [0,1]$, and $F(s) = {\bf \eps}$ for $s \geq 1$. This gives $G(s) = 0$ for $s\leq 0$, $G(s) =(s^2/2) \, {\bf \eps}$ for $s \in [0,1]$, $G(s) = (s-1/2) \, {\bf \eps}$ 
for $s \geq 1$. Another example is a pulse with modulated circular polarization. If the  polarization is perpendicular to the $z$-axis, such a pulse is described by \[ f(s) = h(s) (\cos (\omega (s- 1/2)), \sin (\omega (s- 1/2)), 0) \,  \] where $h(s) \geq 0$ is symmetric decreasing about $s= 1/2$, with $\supp h \subset [0,1]$. If the effect of the pulse does not average out to zero, it is simple to check that, in this case, too, the conditions (\ref{eq:ass1}) and (\ref{eq:ass2}) are satisfied; see Sect. \ref{sec:intro}.

\subsection{The potential}

To describe the coupling of the electron to the nucleus, we consider a static potential $V(x)$. We distinguish two sets of assumptions on the potential $V$.

\medskip

{\it Short range potential.} We assume that there is a constant $V_0$ (with $[V_0] =  length^{-1}$), a length scale $D>0$, and an $\alpha >0$ such that
\begin{equation}\label{eq:pot-sr}
|V(x)| \leq \frac{V_0 D}{|x|} \, \frac{1}{(1+(x/D)^2)^{\alpha/2}}
\end{equation}

{F}rom the physical point of view, it is important to also cover an attractive Coulomb potential.

\medskip

{\it Coulomb potential.} \begin{equation}\label{eq:pot3} V(x) = - \frac{Z}{|x|}, \quad Z >0  \, . \end{equation}

\subsection{The initial wave function}

We require exponential decay of the wave function $\psi$ and of its first and second derivatives. In other words, we assume that 
\begin{equation}\label{eq:psi1} |\psi (x)| \leq C R^{-3/2} \, e^{-|x|/R}, \quad  |\nabla \psi (x)| \leq C R^{-5/2} e^{-|x|/R}, \quad |\Delta \psi (x)| \leq C R^{-5/2} \, e^{-|x|/R}(R^{-1} + |x|^{-1})  \end{equation}
for some $R>0$ and some dimensionless constants $C$.

\medskip

Moreover, we will also need decay in momentum space. We assume that \begin{equation}\label{eq:psi2} |\wh{\psi} (p)| \leq \frac{C R^{3/2}}{(1+ (Rp)^2)^{\gamma/2}} \end{equation}
for a dimensionless constant $C$, and for some $\gamma > 3/2$.

\subsection{The observable} 

For fixed $\delta, \theta >0$, we are interested in the probability that the electron is ejected in the direction $G(t/T)$ of the pulse (with $G(t/T) \to F (1)$, as $t/T \to \infty$). To this end, we propose to estimate the norm \[ N(t)= \| \chi_{\delta,\theta} (t) \cU (t,0) \psi \|,\] where the propagator $\cU (t,0)$ is defined in (\ref{eq:ev}), and 
\[ \chi_{\delta ,\theta} (t) = {\bf 1} (|x| \geq \delta t) \, {\bf 1} (x \cdot G(t/T) \geq |x| |G(t/T)| \cos \theta) \]
for some fixed positive $\delta,\theta$, with $\theta >0$ arbitrarily small. We will prove that if the dimensionless quantity $R\lambda$ is sufficiently large the norm $N(t)$ can be made arbitrarily close to one. Note that our results are {\it uniform} in time $t$. In particular, they hold in the limit of large $t/T$. We observe that, for large $t/T$, the direction of $G(t/T)$ approaches the direction of $F(1)$; in other words, the vector $F(1)$ determines the direction in which the electron propagates asymptotically, after ionization, in the limit of large $R\lambda$.  

\section{Results and Proofs}
\setcounter{equation}{0}

\subsection{Short range potentials}
\label{sec:sr}

We begin our analysis by considering an interaction potential decaying faster than Coulomb. That is, we assume, in this subsection, that $V$ satisfies condition (\ref{eq:pot-sr}), for some $\alpha >0$. 

{\it Notation.} Throughout the paper, $C$ will denote a universal constant, {\it independent} of the parameters $\lambda,T, R, D, V_0$ characterizing the pulse, the initial wave function, and the interaction potential.

{\it Remark.} Note that,  with our conventions, $[T]=[D]=[R]=length$, $[V_0]=[\lambda]=length^{-1}$, and $Z=Ze^2$ is \emph{dimensionless}. We have chosen the numerical value $m_{el}=1/2$ for the electron mass. Therefore, in the formulae below, $t$ stands for $t/2m_{el}$, $T$ stands for $T/2m_{el}$, $\delta$ stands for $2\delta m_{el}$, $V_0$ stands for $2V_0m_{el}$, and (in Section \ref{sec:cou}) $Z$ stands for $2 Z m_{el}$.

\begin{theorem}\label{thm:sr} Assume that conditions (\ref{eq:ass0}), (\ref{eq:ass1}), (\ref{eq:ass2}), (\ref{eq:pot-sr}) for some $\alpha >0$, (\ref{eq:psi1}), (\ref{eq:psi2}) for some $\gamma > 5/2$, are satisfied. Then we have that, uniformly in $t \geq T$,
\[ \begin{split}  \| \chi_{\delta,\theta} (t) \, \cU(t,0) \psi \| \geq \; &1 - C \left[ \frac{1}{R (C\lambda - \delta)} + \frac{1}{R \lambda \tan \theta} \right] \, \left[ 1+ \frac{R^2}{t} \right] \\ &-
\frac{C \,V_0 T}{\alpha (\lambda T/D)^{1+\alpha}} \, \left[ 1+ \frac{R^4}{T^2} \right] 
- CV_0 D R \left[ 1 + \frac{R^{4}}{T^2} \right] \, \kappa_{\lambda}
\end{split} \] where the dimensionless quantity $\kappa_\lambda$ is given by 
\begin{equation}\label{eq:kappa}   \kappa_\lambda = \inf_{0< s_0 <1}  \left\{ \frac{T}{R^2} \, s_0 + \frac{1}{R\lambda} \int_{s_0}^{1} \frac{\rd \tau}{|G(\tau)|} 
\left[ 1 + \frac{1}{\tau^2} \right] \right\} . \end{equation} 
\end{theorem}

{\it Remark.} It follows from assumption (\ref{eq:ass0}) that $\kappa_\lambda \to 0$, as $R\lambda \to \infty$. In fact, it follows from (\ref{eq:ass0}) that the function $K(s_0) = \int_{s_0}^1 \rd \tau |G(\tau)|^{-1} (1+\tau^{-2})$ is finite, for all $s_0 >0$. Clearly, $K(s_0)$ is monotonically decreasing in $s_0$ and can therefore be inverted. Typically $K(s_0) \to \infty$, as $s_0 \to 0$. 
Thus, for $R\lambda$ large enough, we can choose $s_0 = K^{-1} ((R\lambda)^{1/2})$. Then $s_0 \to 0$ and $K(s_0) (\lambda R)^{-1} \to 0$, as $R\lambda \to \infty$. To have more precise information about how fast $\kappa_\lambda$ tends to zero, as $R\lambda \to \infty$, one needs more information about the pulse.

\medskip

{\it Example:} If $f(s) = \eps$, for a fixed $\eps \in \bR^3$, for all $s \in [0,1]$, and $f(s) = 0$ for $s \not \in [0,1]$, it follows that $F(s) = s \, \eps$ and $G(s) = (s^2/2) \eps$, for all $s \in [0,1]$. Then we find that 
\[ \begin{split} \kappa_{\lambda} =\; & \, \inf_{s_0 \in (0,1)} \left\{ \frac{T}{R^2} s_0 + \frac{2}{R\lambda} \int_{s_0}^{1} \frac{\rd \tau}{\tau^{2}} \left[ 1+ \frac{1}{\tau^2}  \right]\right\} \\ \leq \; &  \, \inf_{s_0 \in (0,1)} \left\{ \frac{T}{R^2} s_0 + \frac{2}{R\lambda s_0} + \frac{2}{3R\lambda s_0^3} \right\} . \end{split} \]
It is easy to check that the infimum is attained at $s_0^2 = (R/T\lambda) (1+ \sqrt{1+2T\lambda/R})$. For $R\lambda \gg 1$ (and $R^2 /T \simeq 1$), the infimum is attained at $t_0 \simeq (2TR^5/\lambda)^{1/4}$ and is given by \[ \kappa_\lambda \simeq \frac{4}{3} \, \left( \frac{2T^3}{R^7\lambda} \right)^{1/4} \, . \] 

{\it Remark.} It follows from Theorem \ref{thm:sr}, that, as $t \to \infty$, the electron will propagate, with probability approaching one, as $R\lambda \to \infty$, into the cone with an opening angle smaller than an arbitrary $\theta >0$ around the direction of $F(1)$. In other words,
with $\wt\chi_{\delta,\theta} (t) = {\bf 1} (|x| \geq \delta t) \, {\bf 1} (x \cdot F(1) \geq |x| |F(1)| \cos \theta)$, we find 
\begin{equation}\label{eq:wtchi} \begin{split} \lim\inf_{t\to \infty} \left\| \wt\chi_{\delta, \theta} (t)  \cU(t,0) \psi \right\| \geq \; &1 - C \left[ \frac{1}{R (C\lambda - \delta)} + \frac{1}{R \lambda \tan \theta} \right] \, \\ &- \frac{C V_0 T}{(\lambda T/D)^{1+\alpha}} \, \left[ 1+ \frac{R^4}{T^2} \right]  - CV_0 D R \left[ 1 + \frac{R^{4}}{T^2} \right] \, \kappa_{\lambda}
\end{split} \end{equation}
To prove (\ref{eq:wtchi}), observe that $\| {\bf 1} (x \cdot G(t/T) \geq |x| |G(t/T)| \cos \theta) \psi \| \geq \|  {\bf 1} (x \cdot F(1) \geq |x| |F(1)| \cos (\theta/2)) \psi \|$, if the angle between $G(t/T)$ and $F(1)$ is smaller than $\theta/2$. Since $G(t/T) = G(1) + (t/T-1) F(1)$, the angle between $G(t/T)$ and $F(1)$ is certainly smaller than $\theta/2$, for sufficiently large $t/T \gg 1$. 

\bigskip

To prove Theorem \ref{thm:sr}, we first show how the evolution up to time $T$ can be approximated by the evolution generated by the time dependent Kramers Hamiltonian without potential. The next lemma is due to Fring, Kostrykin and Schrader; see \cite{FKS}. 
\begin{lemma}\label{lm:FKS}
Let $\cU_0 (t,s) = e^{-i\int_s^t \rd \tau \, (p-A(\tau))^2}$. Assume that conditions (\ref{eq:ass1}), (\ref{eq:ass2}), (\ref{eq:psi1}), (\ref{eq:psi2}), and (\ref{eq:pot-sr}), for some $\alpha \geq 0$, are satisfied; (for $\alpha =0$ we recover the Coulomb potential (\ref{eq:pot3}) by taking $V_0 D = Z$). Then there exists a constant $C$ such that
\[ \left\| \left( \cU (T,0) - \cU_0 (T,0) \right) \psi \right\| \leq C V_0 D R \left[ 1 + \frac{R^{4}}{T^2} \right] \, \kappa_{\lambda} \]
with $\kappa_\lambda$ defined in (\ref{eq:kappa}). 
\end{lemma}

\begin{proof}
We define the new propagator \[ \wt{\cU} (s,0) = e^{-2ip \cdot \int_0^s \rd \tau \, A(\tau)} \, e^{i\int_0^s \rd \tau A^2 (\tau)} \, \cU (s,0) . \] Then \[ i \frac{d}{ds} \wt{\cU} (s,0) = \wt{H} (s) \wt{\cU} (s,0)  \] with \[ \wt{H} (s) = p^2 + V (x - 2\lambda T G(s/T))\,. \] Since \[ \cU (T,0) - \cU_0 (T,0) = e^{2i\lambda T  p \cdot G(1)} e^{-i\int_0^T \rd\tau \, A^2 (\tau)} \, \left( \wt{\cU} (T,0) - e^{-iT p^2} \right) \] we find that
\begin{equation} \left\| \left( \cU (T,0) - \cU_0 (T,0) \right) \psi \right\| = \left\| \left( \wt{\cU} (T,0) - e^{-iTp^2} \right) \psi \right\| \leq \int_0^T \rd s \, \left\| V(x-2\lambda T G(s/T)) e^{-isp^2} \psi \right\| \,.\end{equation}

Now, we observe that, on one hand, by (\ref{eq:pot-sr}), 
\begin{equation}\label{eq:H1} \| V(x-2\lambda T G(s/T)) e^{-isp^2} \psi \| \leq 2 V_0 D   \, \int \rd x |\nabla \psi (x)|^2 \leq C \frac{V_0 D}{R}\,.
\end{equation}
On the other hand, by Lemma \ref{lm:V1} (see (\ref{eq:lm3}) below), we have that
\[ \| V(x-2\lambda T G(s/T)) e^{-isp^2} \psi \| \leq \frac{C V_0 D}{\lambda T \, |G(s/T)|}   \left[ 1 + \frac{R^{4}}{s^2} \right] \,. \]
(We are neglecting here the factor $(1+\lambda^2 T^2 |G(s/T)|^2 / D^2)^{-\alpha/2}$ on the r.h.s. of (\ref{eq:lm3}). This factor will play an important role for large times; here it would just give a faster decay in $\lambda$.) Thus
\begin{equation}\begin{split} \left\| \left( \cU (T,0) - \cU_0 (T,0) \right) \psi \right\| \leq &\; \int_0^{t_0} \rd s \, \frac{C V_0 D}{R} + \int_{t_0}^T \rd s \frac{C V_0 D}{\lambda T \, |G(s/T)|}  
\left[ 1 + \frac{R^{4}}{s^2} \right],  \end{split}
\end{equation} for arbitrary $t_0 \in [0,T]$, and hence
\[ \left\| \left( \cU (T,0) - \cU_0 (T,0) \right) \psi \right\| \leq C V_0 D \left[ 1 + \frac{R^{4}}{T^2} \right] \,  \inf_{0< t_0 <T} \left\{\frac{t_0}{R} + \frac{1}{\lambda} \int_{t_0/T}^{1} \frac{\rd \tau}{|G(\tau)|} 
\left[ 1 + \frac{1}{\tau^2} \right] \right\}. \]
\end{proof}

\begin{proof}[Proof of Theorem \ref{thm:sr}] We begin by writing \[ \begin{split} \chi_{\delta,\theta} (t) \, \cU (t,0) \psi =\;& \chi_{\delta, \theta} (t) \, \cU (t,T) \, \cU(T,0) \psi \\ = \; &\chi_{\delta,\theta} (t) \,\cU (t,T) \, \left( \cU(T,0) - e^{-i\int_0^T \rd \tau \, (p-A(\tau))^2} \right) \psi \\ &+ \chi_{\delta,\theta} (t) \, \cU(t,T)  \, e^{-i\int_0^T \rd\tau \, (p-A(\tau))^2} \psi \,. \end{split} \]
Therefore, by Lemma \ref{lm:FKS},
\[ \begin{split} \| \chi_{\delta,\theta} (t)  \, \cU (t,0) \psi \| \geq \; & \| \chi_{\delta,\theta} (t) \, \cU (t,T) e^{-i\int_0^T \rd \tau \, (p-A(\tau))^2} \psi \| - \| ( \cU(T,0) - e^{-i\int_0^T \rd\tau \, (p-A(\tau))^2} ) \psi \| \\ \geq \; & \| \chi_{\delta,\theta}(t) \, \cU (t,T) e^{-i\int_0^T \rd\tau (p-A(\tau))^2} \psi \|  - CV_0 D R \left[ 1 + \frac{R^{4}}{T^2} \right] \, \kappa_{\lambda}  ,\end{split} \]
with $\kappa_\lambda$ defined in (\ref{eq:kappa}).
Since $A(t) = A(T)$, for all $t >T$, we obtain that
\begin{equation}\label{eq:1}
 \| \chi_{\delta,\theta} (t) \, \cU (t,0) \psi \| \geq \| \chi_{\delta,\theta} (t) \, e^{-i(t-T) \left[ (p-A(T))^2 + V(x) \right]} \,  e^{-i\int_0^T \rd \tau \, (p-A(\tau))^2} \psi \|  - C V_0 D R \left[ 1 + \frac{R^{4}}{T^2} \right] \, \kappa_{\lambda}\,.\end{equation}
Next, we notice that
\begin{equation}\label{eq:intV0}
\begin{split}
\| \chi_{\delta,\theta} (t) \, & e^{-i(t-T) \left[ (p-A(T))^2 + V(x) \right]} \,  e^{-i\int_0^T \rd\tau \,(p-A(\tau))^2} \psi \| \\ \geq \; & \| \chi_{\delta,\theta} (t) \,  e^{-i\int_0^t \rd\tau \, (p-A(\tau))^2} \psi \| \\ &- \left\| \left( e^{-i(t-T) \left[ (p-A(T))^2 + V(x) \right]} - e^{-i(t-T) (p-A(T))^2} \right) \,  e^{-i \int_0^T \rd\tau  (p-A(\tau))^2} \psi \right\| \\
\geq \; & \| \chi_{\delta,\theta} (t) \,  e^{-i\int_0^t \rd\tau \, (p-A(\tau))^2} \psi \| - \int_0^{t-T} \rd s \left\| V(x) e^{-i \int_0^{T+s} \rd \tau (p-A(\tau))^2} \psi \right\|  \,.
\end{split}\end{equation}
We then use that
\begin{equation}\label{eq:intV1}
\begin{split}
\int_0^{t-T} \rd s \left\| V(x) e^{-i \int_0^{T+s} \rd \tau (p-A(\tau))^2} \psi \right\|  = \; &
\int_0^{t-T} \rd s \left\| V(x) e^{-i (T+s) p^2} e^{2 i \lambda T p \cdot G(1+s/T)}\psi \right\| \\ = \; & \int_0^{t-T} \rd s \left\| V(x-2\lambda T G(1+s/T)) e^{-i (T+s) p^2} \psi \right\|\,.
\end{split}\end{equation}
To bound the integrand, we observe that, 
by Lemma \ref{lm:V1} (see (\ref{eq:lm3}) below),
\[ \begin{split} \left\| V(x-2\lambda T G(1+s/T)) e^{-i (T+s) p^2} \psi \right\| \leq 
\frac{C V_0}{(\lambda T  |G(1+s/T)| / D)^{1+\alpha}}   \left[ 1 + \frac{R^{4}}{T^2} \right]
\end{split} \]
for all $s \geq 0$. Hence, by (\ref{eq:ass2}),
\[ \begin{split} 
\int_0^{t-T} \rd s \, & \left\| V(x-2\lambda T G(1+s/T)) e^{-i (T+s) p^2} \psi \right\| \\ &\hspace{2cm} \leq \;  \frac{C V_0}{(\lambda T / D)^{1+\alpha}} \, \left[ 1+ \frac{R^4}{T^2}  \right] \int_{0}^{t-T} \rd s \, \frac{1}{|G(1+s/T)|^{1+\alpha}} \\ & \hspace{2cm}
 \leq \;  \frac{C V_0 T}{(\lambda T / D)^{1+\alpha}} \, \left[ 1+ \frac{R^4}{T^2}  \right] \int_{0}^{\infty} \rd \tau \, \frac{1}{(1+\tau)^{1+\alpha}} 
 \\ & \hspace{2cm}
 \leq \;  \frac{C V_0 T}{\alpha (\lambda T / D)^{1+\alpha}} \, \left[ 1+ \frac{R^4}{T^2}  \right] .
\end{split} \,.\]
 
Therefore, from (\ref{eq:intV0}), 
\begin{equation}\label{eq:concl0} \begin{split} \| \chi_{\delta,\theta} (t) \, & e^{-i(t-T) \left[ (p-A(T))^2 + V(x) \right]} \,  e^{-i\int_0^T \rd\tau \,(p-A(\tau))^2} \psi \| \\ & \geq  \| \chi_{\delta,\theta} (t) \,  e^{-i\int_0^t \rd\tau \, (p-A(\tau))^2} \psi \| - 
\frac{C V_0 T}{\alpha  \, (\lambda T/D)^{1+\alpha}} \, \left[ 1+ \frac{R^4}{T^2} \right] \,. \end{split}\end{equation}
The first term on the right hand side of (\ref{eq:concl0}) can be bounded by
\begin{equation}\label{eq:concl1}
\begin{split} \| \chi_{\delta,\theta} (t) \, &e^{-i\int_0^t \rd\tau \, (p-A(\tau))^2 } \psi \| \\ \geq \; & 1 - \| {\bf 1} (|x| \leq \delta t) e^{2i\lambda T p \cdot G(t/T)} \, e^{-it p^2} \psi \| \\ &- \| {\bf 1} (x \cdot G(t/T) \leq |x| |G(t/T)| \cos \theta)  e^{2i\lambda T p \cdot G(t/T)} \, e^{-it p^2} \psi \| \\
\geq \; & 1 - \| {\bf 1} (|x-2\lambda T G(t/T)| \leq \delta t) \, e^{-it p^2} \psi \| \\ &- \| {\bf 1} ((x-2\lambda T G(t/T)) \cdot G(t/T) \leq |x-2\lambda T G(t/T)| |G(t/T)| \cos \theta)  \, e^{-it p^2} \psi \| \,.
\end{split} \end{equation}
Since, by (\ref{eq:ass2}), $|G(s)| \geq C s$ for all $s \geq 1$, we find that
\begin{equation}\label{eq:concl2} \begin{split} \| \chi_{\delta,\theta} (t) \, e^{-i\int_0^t \rd\tau \, (p-A(\tau))^2 } \psi \| \geq 
1 - \| {\bf 1} (|x| \geq (C\lambda -\delta) t) \, e^{-it p^2} \psi \| - \| {\bf 1} (|x| \geq C\lambda t \tan \theta) e^{-it p^2} \psi \| \,.
\end{split} \end{equation}
To conclude, we observe that
\[ \begin{split} \| {\bf 1} (|x| \geq K t) \, e^{-it p^2} \psi \|^2 \leq \; & \frac{1}{(Kt)^2} \langle e^{-itp^2} \psi, x^2 e^{-itp^2} \psi \rangle \leq  \frac{1}{(Kt)^2} \langle \psi, (x+2tp)^2 \psi \rangle  \\ \leq \; & \frac{2}{(Kt)^2} \langle \psi, (x^2 +4 t^2 p^2) \psi \rangle \leq \frac{C}{(Kt)^2} (R^2 + t^2 R^{-2}) \leq \frac{C}{(KR)^2} \left(1+\frac{R^4}{t^2} \right) \end{split}\] 
using (\ref{eq:psi1}), and (\ref{eq:psi2}) for some $\gamma >5/2$. Hence, (\ref{eq:concl2}) yields 
\[  \| \chi_{\delta,\theta} (t) \, e^{-i\int_0^t \rd\tau \, (p-A(\tau))^2 } \psi \| \geq 
1 - C \left[ \frac{1}{R (C\lambda - \delta)} + \frac{1}{R \lambda \tan \theta} \right] \, \left[ 1+ \frac{R^2}{t} \right] \,. \]
Together with (\ref{eq:1}) and (\ref{eq:concl0}), this concludes the proof of the theorem. 
\end{proof}

\begin{lemma}\label{lm:V1}
Assume (\ref{eq:pot-sr}), for some $\alpha >0$, and (\ref{eq:psi1}), (\ref{eq:psi2}), for some $\gamma >3/2$. Then
\begin{equation}\label{eq:lm3} \| V(x-2\lambda T G(t/T)) e^{-itp^2} \psi \| \leq \frac{C V_0 D}{\lambda T \, |G(t/T)|}  \frac{1}{(1+  \frac{\lambda^2 T^2}{D^2} |G(t/T)|^2)^{\alpha/2}} \left[ 1 + \frac{R^{4}}{t^2} \right]\,. \end{equation}
\end{lemma}
\begin{proof}
We notice that
\[ \begin{split}
(e^{-it p^2} \psi) (x)= \; &\frac{e^{ix^2/4t}}{(4\pi i t)^{3/2}} \int \rd y \, e^{-iy \cdot x/2t} e^{iy^2/4t} \psi (y) \\ = \; &
\frac{e^{ix^2/4t}}{(4\pi i t)^{3/2}} \wh{\psi} (x/2t) + \frac{e^{ix^2/4t}}{(4\pi i t)^{3/2}} \int \rd y \, e^{-iy \cdot x/2t} \left( e^{iy^2/4t} - 1\right) \, \psi (y) \,.
\end{split}
\]
In the second term, we perform integration by parts to obtain decay in the $x$-variable.
\[ \begin{split}
\int \rd y \, &e^{-iy \cdot x/2t} \left( e^{iy^2/4t} - 1\right) \, \psi (y) \\ = \; &-\int \rd y \, \frac{\Delta_y e^{-iy \cdot x/2t}}{(x/2t)^2} \left( e^{i y^2/ 4t} - 1 \right) \psi (y) \\ 
= \; &-\frac{1}{(x/2t)^2} \, \int \rd y \, e^{-iy \cdot x/2t} \\ &\times \left[ (\Delta \psi) (y) \,  \left( e^{i y^2/ 4t} - 1 \right) + i (\nabla \psi) (y) \cdot \frac{y}{2t} \,  e^{i y^2/4t} + \psi (y) \left(-\frac{y^2}{4t^2} +  \frac{3i}{2t}\right) e^{iy^2/4t} \right] \, .\end{split} \]
Therefore, we obtain that
\[ \begin{split} \Big| \int \rd y \, e^{-iy \cdot x/2t} &\left( e^{iy^2/4t} - 1\right) \, \psi (y) \Big| \\  \leq & \; \frac{1}{(x/2t)^2}  \int \rd y \, \left\{ |\Delta \psi (y)| \frac{|y|^2}{4t} + |\nabla \psi (y)| \frac{|y|}{t} + |\psi (y)| \left(\frac{3}{2t} + \frac{|y|^2}{4t^2} \right) \right\} \\ \leq &  \; 
C \; \frac{R^{3/2}}{t} \, \frac{1}{(x/t)^2} \left( 1 + \frac{R^2}{t} \right) \, .\end{split} \]
Since, on the other hand,
\[ \left| \int \rd y \, e^{-iy \cdot x/2t} \left( e^{iy^2/4t} - 1\right) \, \psi (y) \right|  \leq C R^{3/2} \frac{R^2}{t} \, , \] 
it follows that
\[ \left| \int \rd y \, e^{-iy \cdot x/2t} \left( e^{iy^2/4t} - 1\right) \, \psi (y) \right|  \leq 
C \, R^{3/2} \, \frac{R^{2}}{t} \frac{1}{1 + (Rx/t)^2}  \left(1 + \frac{R^2}{t} \right)\,. \] 
Hence, by using (\ref{eq:pot-sr}) and (\ref{eq:psi2}), 
\[ \begin{split} 
 \Big\| V(x-2\lambda &T G(t/T)) e^{-i tp^2} \psi \Big\|^2 \\ \leq \; &  C \int \rd x \, V^2 (x-2\lambda T G(t/T)) \frac{|\wh{\psi} (x/2t)|^2}{t^3} \\ &+ C \, R^{3} \, \frac{R^{4}}{t^2} \left(1+\frac{R^2}{t} \right)^2
  \int \frac{\rd x}{t^3} \, V^2 (x-2\lambda T G(t/T))  \frac{1}{(1 + (Rx/t)^2)^2}
\\ \leq \; &  \frac{C V^2_0 D^{2}}{t^2} \int \frac{\rd x}{|x-\lambda (T/t) G(t/T)|^{2}} \, \frac{1}{(1+4 t^2 D^{-2} |x - \lambda (T/t) G(t/T)|^2)^{\alpha}} |\wh{\psi} (x)|^2 \\ &+ \frac{C V^2_0 D^{2}}{t^{2}} \, 
    R^{3} \, \frac{R^{4}}{t^2} \left(1+\frac{R^2}{t} \right)^2  \\ &\times \int \, \frac{\rd x}{|x-\lambda (T/t) G(T/t)|^{2}} \frac{1}{(1+4 t^2 D^{-2} |x- \lambda (T/t) G(T/t)|^2)^{\alpha}}  \frac{1}{(1 + (Rx)^2)^2} \\ \leq \; &\frac{C V^2_0 D^{2} R^2}{t^2} \left[ 1 + \frac{R^{4}}{t^2} \left(1+\frac{R^2}{t} \right)^2 \right]  \\ &\times \int \, \frac{\rd x}{|x-2 R\lambda (T/t) G(t/T)|^{2}} \frac{1}{(1+ 4t^2 R^{-2} D^{-2} |x- 2 R\lambda (T/t) G(t/T)|^2)^{\alpha}}  \frac{1}{(1 + x^2)^\beta} \end{split} \]
where $\beta = \min (\gamma, 2) >3/2$. It follows that
\[ \begin{split}
 \left\| V(x-2\lambda T G(t/T)) e^{-i tp^2} \psi \right\|^2 \leq
      \; &\frac{C V^2_0 D^{2}}{\lambda^2 T^2 |G(t/T)|^2}  \frac{1}{(1+  \frac{\lambda^2 T^2}{D^2} |G(t/T)|^2)^{\alpha}} \left[ 1 + \frac{R^{4}}{t^2} \right]^2 \,. \end{split} \]
\end{proof}

\subsection{Coulomb potentials}
\label{sec:cou}

In this section we consider the physically more interesting case of a Coulomb interaction. 
The long range of the Coulomb potential requires some modification of the argument used in the previous section; in particular, to obtain results uniform in time, we need to approximate the long time evolution by a ``Dollard-modified'' free dynamics (see \cite{10}).

As initial data we consider here the ground state of the Schr\"odinger operator with an attractive Coulomb interaction, which satisfies the assumptions (\ref{eq:psi1}), and (\ref{eq:psi2}), with $\gamma =4$. (In the following theorem we will therefore assume (\ref{eq:psi2}) with $\gamma=4$; but, of course, other values of $\gamma$ can also be considered.)

\begin{theorem}\label{thm:cou} Assume that conditions (\ref{eq:ass1}), (\ref{eq:ass2}), (\ref{eq:pot3}), and (\ref{eq:psi2}), for $\gamma = 4$, are satisfied. Suppose that there exists a constant $C$ such that $C^{-1} \leq R^2/ T \leq C$, that $Z \leq \lambda$, and that $\lambda R \geq 1$ is large enough. Then we have that, uniformly in $t \geq T$, 
\[ \begin{split} \| \chi_{\delta,\theta} (t) \, \cU(t,0) \psi \| \geq \; & 
1 -  C \left(\frac{1}{R\lambda \tan\theta} + \frac{1}{R (C\lambda-\delta)}\right)  \left(1+\frac{R^2}{t} \right) \\ &-  Z R \left[ 1 + \frac{R^{4}}{T^2} \right] \, \kappa_{\lambda} 
  - \frac{C}{(R\lambda)^{1/7}}  \left( \frac{Z T^{3/2}}{R^2} \right)^{4/7} \end{split} \]
where the dimensionless quantity $\kappa_\lambda$ was defined in (\ref{eq:kappa}). Since, by assumption (\ref{eq:ass0}), $\kappa_\lambda \to 0$, as $(\lambda R) \to \infty$, it follows in particular that
\[ \lim_{\lambda R \to \infty} \| \chi_{\delta,\theta} (t) \, \cU(t,0) \psi \| = 1 \] uniformly in $t \geq T$.
\end{theorem}

{\it Remark.} Just like Theorem \ref{thm:sr}, Theorem \ref{thm:cou} implies that \[\begin{split}  \lim\inf_{t \to \infty} \left\|  \wt{\chi}_{\delta,\theta} (t) \cU (t,0) \psi \right\| \geq &1- C \left(\frac{1}{R\lambda \tan\theta} + \frac{1}{R (C\lambda-\delta)}\right)  \\ &-  Z R \left[ 1 + \frac{R^{4}}{T^2} \right] \, \kappa_{\lambda}  - \frac{C}{(R\lambda)^{1/7}}  \left( \frac{Z T^{3/2}}{R^2} \right)^{4/7} \end{split} \]
where $\wt{\chi}_{\delta,\theta} (t) = {\bf 1} (|x| \geq t \delta) {\bf 1} (x\cdot F(1) \geq |x| |F(1)| \cos \theta)$. In other words, it is the vector $F(1)$ that determines, with probability approaching one, as $R\lambda \to \infty$, the direction in which the electron propagates asymptotically.

\begin{proof}
By Lemma \ref{lm:FKS} we have that
\begin{equation}\label{eq:cou0}
\| \chi_{\delta,\theta} (t) \, \cU(t,0) \psi \| \geq \| \chi_{\delta,\theta} (t) \, \cU(t,T) e^{-i\int_0^T \rd \tau (p-A(\tau))^2} \psi \| - C Z R \left[ 1 + \frac{R^{4}}{T^2} \right] \, \kappa_{\lambda}\, .
\end{equation}

In order to replace the unitary evolution $\cU(t,T)$ by a free evolution, we introduce, first of all, a cutoff in momentum space. We choose a smooth function $\chi \in C_0^{\infty} (\bR^3)$, with $\chi (x) = 0$ for all $|x| \geq 1$ and $\chi (x) = 1$ for all $|x| \leq 1/2$. We define $\bar{\chi} = 1 -\chi$. Then we have
\begin{equation}\label{eq:int-cou0}
\begin{split}
\| \chi_{\delta,\theta} (t) \, \cU(t,T) &e^{-i\int_0^T \rd \tau (p-A(\tau))^2} \psi \| \\ \geq \; & \| \chi_{\delta,\theta} (t) \, \cU(t,T) e^{-i\int_0^T \rd \tau (p-A(\tau))^2} \chi (p/K_0) \psi \| - \| \bar{\chi} (p/K_0) \psi \| \\ \geq \; & \| \chi_{\delta,\theta} (t) \, \cU(t,T) e^{-i\int_0^T \rd \tau (p-A(\tau))^2} \chi (p/K_0) \psi \| - \frac{C}{(RK_0)^{\gamma-3/2}} 
\end{split}
\end{equation}
for arbitrary $K_0 >0$; we will later optimize the choice of $K_0$. 
Next, we let $\psi_T = e^{-iTp^2} \chi (p/K_0) \psi$, and we observe that
\begin{equation}\label{eq:doll}
\begin{split}
\chi_{\delta,\theta} &(t) \, \cU(t,T) e^{-i\int_0^T \rd s (p-A(s))^2} \chi (p/K_0) \psi \\ = 
\; & \chi_{\delta,\theta} (t) e^{-i(t-T) \left[ (p-A(T))^2 - Z /|x| \right]} e^{2i\lambda T p \cdot G(1)} e^{-i\int_0^T \rd s \,A^2 (s)} \psi_T \\ = \; & e^{ix \cdot A(T)} \chi_{\delta,\theta} (t) \,  e^{-i(t-T) \left[ p^2 - Z/ |x| \right]} e^{-ix \cdot A(T)}  e^{2i\lambda T p \cdot G(1)} e^{-i\int_0^T \rd s \, A^2 (s)} \psi_T \\ = \; & e^{-i\int_0^T\rd s A^2 (s)} e^{2i \lambda T G(1) \cdot A(T)} \, e^{ix \cdot A(T)} \chi_{\delta,\theta} (t) \, e^{2i\lambda T p \cdot G(1)} e^{-i(t-T) \left[ p^2 -Z/ |x-2\lambda T G(1)| \right]} e^{-i x \cdot A(T)}  \psi_T 
\end{split}\end{equation}
and we write
\begin{equation}\label{eq:doll2}
\begin{split}
&e^{-i(t-T)  \left[ p^2 - Z / |x-2\lambda T G(1)| \right]}  e^{-i x \cdot A(T)} \psi_T 
\\ & \hspace{2cm} =  e^{-i(t-T)p^2} e^{i Z  \int_0^{t-T} \frac{\rd \tau}{|2 \tau p - 2 \lambda T G(1)|}} e^{-i x \cdot A(T)} \psi_T  \\ &\hspace{2.5cm} +\left[ e^{-i(t-T) \left[ p^2 - Z / |x-2\lambda TG(1)| \right]} -  e^{-i(t-T)p^2} e^{i Z  \int_0^{t-T} \frac{\rd \tau}{|2 \tau p - 2 \lambda T G(1)|}}\right] e^{-i x \cdot A(T)} \psi_T \,.
\end{split}
\end{equation}

Observe here that $\psi_T = \chi (p/K_0) e^{-iTp^2}\psi$ is supported, in momentum space, in the ball of radius $K_0$ around the origin. This implies that $|p+\lambda F(1)| \leq K_0$ for all $p$ in the support of the Fourier transform of $e^{-ix\cdot A(T)} \psi_T$. Therefore $|2\tau p - 2 \lambda T G(1)| \geq 2 \lambda T |G(1+\tau/T)| - 2 \tau K_0 \geq C \lambda  T +  (C \lambda -K_0) \tau$ for all $\tau \in [0, t-T]$. In particular, if we require that $K_0 \leq C\lambda/2$, the integral  $\int_0^{t-T} \rd \tau \, |2\tau p - 2 \lambda T G(1)|^{-1}$ is well defined (at the end, we will choose $K_0 R \simeq (\lambda R)^{2/35}$, and therefore the condition $K_0 \leq C \lambda/2$ is certainly satisfied for sufficiently large values of $(\lambda R)$). It follows that
\begin{equation}\label{eq:cou-int0b} \begin{split}
\| \chi_{\delta,\theta} (t) \, &\cU(t,T) e^{i\int_0^T \rd s (p-A(s))^2} \chi (p/K_0) \psi \| \\ \geq\; & \| \chi_{\delta,\theta} (t) \, e^{2i\lambda T p \cdot G(1)} e^{-i(t-T) p^2} e^{i Z  \int_0^{t-T} \frac{\rd \tau}{|2\tau p -2\lambda T G(1)|} } e^{-i x \cdot A(T)} \psi_T \| \\ &-  \left\|  \left[ e^{-i(t-T) [p^2 - Z/ |x-2\lambda T G(1)|]} - e^{-i(t-T) p^2} e^{i Z \int_0^{t-T} \frac{\rd \tau}{|2\tau p -2\lambda T G(1)|} } \right] e^{-i x \cdot A(T)} \psi_T \right\| \,. \end{split}\end{equation}
To bound the first term, we observe that
\begin{equation}\label{eq:doll3} \begin{split}
 \| \chi_{\delta,\theta} (t) \, e^{2i\lambda T p \cdot G(1)} &e^{-i(t-T) p^2} e^{i Z  \int_0^{t-T} \frac{\rd \tau}{|2\tau p - 2 \lambda T G(1)|} } e^{-i x \cdot A(T)} \psi_T \| 
 \\ =\; &  \| \chi_{\delta,\theta} (t) \, e^{2i\lambda T (p-A(T)) \cdot G(1)} e^{-i(t-T) (p-A(T))^2} e^{i Z  \int_0^{t-T} \frac{\rd \tau}{|2\tau (p-A(T)) - 2 \lambda T G(1)|} } \psi_T \| \\ = \;& \| \chi_{\delta,\theta} (t) \, e^{2i \lambda T p \cdot G(t/T)} e^{-i t p^2} e^{i Z  \int_0^{t-T} \frac{\rd \tau}{|2\tau p - 2 \lambda T G(1+\tau/T)|} } \chi (p/K_0) \psi \|  \,.
\end{split}\end{equation}

\medskip

%

{F}rom (\ref{eq:doll3}), we obtain
\[ \begin{split} 
  \| \chi_{\delta,\theta} (t) \, e^{2i\lambda T p \cdot G(1)} &e^{-i(t-T) p^2} e^{i Z  \int_0^{t-T} \frac{\rd \tau}{|2\tau p - 2 \lambda T G(1)|} } e^{-i x \cdot A(T)} \psi_T \| 
  \\ = \;& 
\|  {\bf 1} ((x-2\lambda T G(t/T)) \cdot G(t/T) \geq |x-2\lambda T G(t/T)| |G(t/T)| \cos \theta)  \\ &\times {\bf 1} (|x-2\lambda T G(t/T)| \geq \delta t) e^{-i t p^2} e^{i Z  \int_0^{t-T} \frac{\rd \tau}{|2\tau p -2\lambda T G(1+\tau/T)|}} \chi (p / K_0) \psi \| \\ \geq \; & \| \chi (p /K_0) \psi \| \\ &- \| {\bf 1} (|x-2\lambda T G(t/T)| \leq \delta t) e^{-i t p^2} e^{i Z  \int_0^{t-T} \frac{\rd \tau}{|2\tau p -2\lambda T G(1+\tau/T)|}} \chi (p/ K_0) \psi \| \\ &- \| {\bf 1} ((x-2\lambda T G(t/T)) \cdot G(t/T) \leq |x-2\lambda G(t/T)| |G(t/T)| \cos \theta) \\ &\hspace{2cm} \times  e^{-i t p^2} e^{i Z  \int_0^{t-T} \frac{\rd \tau}{|2\tau p -2\lambda T G(1+\tau/T)|}} \chi (p /K_0) \psi \|
\\ \geq \; & 1 - (R K_0)^{-5/2} \\ &- \| {\bf 1} (|x| \geq (C\lambda -\delta) t) \, e^{-i t p^2} e^{i Z  \int_0^{t-T} \frac{\rd \tau}{|2\tau p -2\lambda T G(1+\tau/T)|}} \chi (p / K_0) \psi \| \\ &- \| {\bf 1} (|x| \geq C \lambda t \tan \theta)  e^{-i t p^2} e^{i Z  \int_0^{t-T} \frac{\rd \tau}{|2\tau p -2\lambda T G(1+\tau/T)|}} \chi (p /K_0) \psi \| \,.
\end{split}\]
{F}rom Lemma \ref{lm:free}, below, we find that 
\begin{equation}\label{eq:fr} \begin{split}
 \| \chi_{\delta,\theta} (t) \, e^{2i\lambda T p \cdot G(1)} &e^{-i(t-T) p^2} e^{i Z  \int_0^{t-T} \frac{\rd \tau}{|2\tau p -2\lambda T G(1)|} } e^{-i x \cdot A(T)} \psi_T \|  \\  \geq \; & 1 - (R K_0)^{-5/2} - C \left(\frac{1}{R\lambda \tan\theta} + \frac{1}{R (C\lambda-\delta)}\right) \left(1+\frac{R^2}{t} \right) \,. \end{split}\end{equation}

As for the second term on the r.h.s. of (\ref{eq:cou-int0b}), we use the bound
\begin{equation}\label{eq:cou-int} \begin{split} 
&\Big\|  \Big[ e^{-i(t-T) [p^2 - Z /|x-2\lambda T G(1)|]} - e^{-i(t-T) p^2} e^{i Z  \int_0^{t-T} \frac{\rd \tau}{|2\tau p -2\lambda T G(1)|} } \Big] e^{-i x \cdot A(T)} \psi_T \Big\| \\ &\hspace{.5cm} \leq Z \int_0^{t-T} \rd s \,  \left\| \left[ \frac{1}{|x-2\lambda T G(1)|} - \frac{1}{|2sp - 2\lambda T G(1)|} \right] e^{-isp^2} e^{i \int_0^s \frac{\rd \tau}{|2\tau p - 2 \lambda T G(1)|}} e^{-i x \cdot A(T)} \psi_T \right\| \, .\end{split}\end{equation}
We first handle small values of $s \in [0,t-T]$. To this end, we observe that
\[ \begin{split} \Big\| \frac{1}{|x-2\lambda T G(1)|} e^{-isp^2} & e^{i Z  \int_0^s \frac{\rd \tau}{|2\tau p - 2 \lambda T G(1)|}} e^{-i x \cdot A(T)} \psi_T \Big\|^2 \\ = \; & \left\| \frac{1}{|x-2\lambda T G(1)|}  e^{-i x \cdot A(T)} \, e^{-is(p-A(T))^2} e^{iZ  \int_0^s \frac{\rd \tau}{|2\tau p - 2 \lambda T G(1+\tau/T)|}} \psi_T \right\|^2 \\ \leq \; &\int \frac{\rd x}{|x-2\lambda T G(1)|^2} \left| e^{-is(p-A(T))^2} e^{i Z \int_0^s \frac{\rd \tau}{|2\tau p - 2 \lambda T G(1+\tau/T)|}} \psi_T (x) \right|^2 \\ \leq \; &4 \int_{|p| \leq K_0} \rd p \, |p|^2 \, | \wh{\psi} (p)|^2 \leq C R^{-2} 
\end{split} \]
using (\ref{eq:psi2}), with $\gamma=4$. 
On the other hand, we have that
\[ \begin{split} \left\|  \frac{1}{|2sp - 2\lambda T G(1)|} e^{-isp^2} e^{i Z \int_0^s \frac{\rd \tau}{|2\tau p - 2 \lambda T G(1)|}} e^{-i x \cdot A(T)} \psi_T \right\| &=  \left\|  \frac{1}{|2sp - 2\lambda T G(1+s/T)|}  \chi (p/K_0) \psi \right\| \\ &\leq \frac{1}{2\lambda T |G(1+s/T)| -  s K_0} \\
&\leq \frac{1}{C \lambda T}  \end{split} \]
for all $s \in [0,t]$, if $K_0 <C\lambda/2$; here we used the assumption (\ref{eq:ass2}). Therefore
\[ \begin{split} \Big\| \left[ \frac{1}{|x-2\lambda G(T)|} -  \frac{1}{|2sp -
2\lambda G(T)|} \right] \, e^{-isp^2}  &e^{i Z  \int_0^s \frac{\rd
\tau}{|2\tau p - 2 \lambda T G(1)|}} e^{-i x \cdot A(T)} \psi_T \Big\| \\
&\hspace{1cm} \leq C \left(\frac{1}{R}  
+ \frac{1}{\lambda T}\right) \leq C R^{-1} \end{split} \] 
assuming $\lambda T \geq R$. In conclusion
\begin{equation}\label{eq:cou-int2} \begin{split}  \int_0^{t-T} &\rd s \,  \left\| \left[ \frac{1}{|x-2\lambda T G(1)|} - \frac{1}{|2sp - 2\lambda T G(1)|} \right] e^{-isp^2} e^{i Z \int_0^s
\frac{\rd \tau}{|2\tau p - 2 \lambda T G(1)|}} e^{-i x \cdot A(T)} \psi_T \right\|
\\ \leq \; & \frac{C t_0}{R} + \int_{t_0}^{t-T}  \rd s \,  \left\| \left[
\frac{1}{|x-2\lambda T G(1)|} - \frac{1}{|2sp - 2\lambda T G(1)|} \right] e^{-isp^2}
e^{i Z \int_0^s \frac{\rd \tau}{|2\tau p - 2 \lambda T G(1)|}} e^{-i x \cdot
A(T)} \psi_T \right\| \,. \end{split}\end{equation}
To estimate the second term, we use the kernel representation
\[ (e^{-isp^2} \psi )(x) = \frac{1}{(4\pi i s)^{3/2}} \int \rd y \, e^{i(x-y)^2/4s} \psi (y) \]
implying that
\[ \begin{split}
\frac{1}{|x-2\lambda T G(1)|} & \left( e^{-isp^2}  e^{i Z  \int_0^s \frac{\rd \tau}{|\tau p - 2 \lambda T G(1)|}} e^{-i x \cdot A(T)} \psi_T \right) (x) \\ = \; & \frac{e^{ix^2/4s}}{(4\pi is)^{3/2}|x-2\lambda T G(1)|} e^{i Z   \int_0^s \frac{\rd \tau}{|\tau x/s - 2 \lambda T G(1)|}} \wh{\psi}_T (x/2s + \lambda F(1)) + R^{(1)}_\lambda (s,x) \, , \end{split}
\] with \[ \begin{split} R^{(1)}_\lambda (s,x) = \; &\frac{e^{ix^2/4s}}{(4\pi i s)^{3/2} \, |x-2\lambda T G(1)|}   \\ &\times \int \rd y \, e^{-iy \cdot x/2s} \left( e^{iy^2/4s} - 1 \right) \left( e^{i Z  \int_0^s \frac{\rd \tau}{|2\tau p - 2 \lambda T G(1)|}} e^{-i x \cdot A(T)} \psi_T \right)(y)  \,. \end{split} \]
Similarly, we notice that
\[ \begin{split} \frac{1}{|2sp - 2 \lambda T G(1)|} &e^{-isp^2}  e^{i Z   \int_0^s \frac{\rd \tau}{|2\tau p - 2 \lambda G(T)|}} e^{-i x \cdot A(T)} \psi_T (x) \\ &= \frac{e^{ix^2/4s}}{(4\pi is)^{3/2}} \frac{1}{|x-2\lambda T G(1)|} e^{i Z  \int_0^s \frac{\rd \tau}{|\tau x/s - 2\lambda T G(1)|}} \wh{\psi}_T (x/2s + \lambda F(1)) + R^{(2)}_{\lambda} (s,x), \end{split} \]
with \[ \begin{split} R^{(2)}_\lambda (s,x) =   \frac{e^{ix^2/4s}}{(4\pi is)^{3/2}}  \int \rd y \, &e^{-iy \cdot x/2s} \left( e^{iy^2/4s} - 1 \right) \\ &\times \left( \frac{1}{|2sp - 2\lambda T G(1)|} \, e^{i Z  \int_0^s \frac{\rd \tau}{|2\tau p - 2 \lambda T G(1)|}} e^{-i x \cdot A(T)} \psi_T \right)(y) \,. \end{split} \] {F}rom (\ref{eq:cou-int}), we find that
\begin{equation} \label{eq:RR}
\begin{split} 
\int_{t_0}^{t-T}  \rd s \,  \Big\| \Big[
\frac{1}{|x-2\lambda T G(1)|} - \frac{1}{|sp - 2\lambda T G(1)|} \Big] e^{-isp^2}
&e^{i Z \int_0^s \frac{\rd \tau}{|2\tau p - 2 T \lambda G(1)|}} e^{-i x \cdot
A(T)} \psi_T \Big\| 
 \\ &\leq \int_{t_0}^{t-T} \rd s \, \left( \| R^{(1)}_\lambda (s,x) \| + \| R^{(2)}_\lambda (s,x) \| \right)\,.
\end{split} \end{equation}
To control the first remainder term, we compute \[ \begin{split}
\left( e^{i Z   \int_0^s \frac{\rd \tau}{|2\tau p - 2 \lambda T G(1)|}} e^{-i x \cdot A(T)} \psi_T \right)(y) = &\; \frac{1}{(2\pi)^{3/2}} \int \rd k \, e^{ik\cdot y} e^{i Z   \int_0^s \frac{\rd \tau}{|2\tau k - 2 \lambda T  G(1)|}} \wh{\psi_T} (k + \lambda F(1)) \\ = &\; \frac{e^{-i\lambda F(1) \cdot y}}{(2\pi)^{3/2}} \int \rd k \, e^{ik\cdot y} e^{i Z  \int_0^s \frac{\rd \tau}{|2\tau k - 2 \lambda T G(1+\tau/T)|}} \wh{\psi_T} (k) \,. \end{split}\] Hence
\[ 
R^{(1)}_\lambda (s,x) = \frac{e^{ix^2/4s}}{(8\pi^2 i s)^{3/2} |x- 2\lambda T G(1)|} \int \rd y \, e^{-i y \cdot (x/2s + \lambda F(1))} \, (e^{iy^2/4s} - 1) \, h_\lambda (s,y) \]
with
\[ h_\lambda (s,y) =  \int \rd k \, e^{ik\cdot y} e^{i Z  \int_0^s \frac{\rd \tau}{|2\tau k - 2 \lambda T G(1+\tau/T)|}} \wh{\psi_T} (k) \,. \] 
In Lemma \ref{lm:h}, below, we show that, for every multi-index $\beta \in \bN^3$,
\[  \left| D^{\beta}_x h_\lambda (s,x) \right| \leq
\frac{ C R^{-3/2}  K_0^{|\beta|}}{1+(x/R)^{2n}} \,\left(\frac{K_0 T}{R}\right)^{2n} \left[ 1 + \frac{Z}{\lambda} \log^{2n} (1+ s/T) \right] \,. \]
Therefore, on the one hand,
\begin{equation}\label{eq:R11} \begin{split} |R^{(1)}_\lambda (s,x)| \leq \;& \frac{C}{s^{5/2} |x-2\lambda G(T)|} \int \rd y \, |y|^2 \, |h_\lambda (s,y)| \\ \leq \; &
\frac{C R^{7/2}}{s^{5/2}\, |x-2\lambda G(T)|} \left( \frac{TK_0}{R} \right)^{2n} \left[ 1 + \frac{Z}{\lambda} \log^{2n} (1+s/T) \right] \end{split} \end{equation}
for all $n > 5/2$. On the other hand, from 
\[ \begin{split} 
R^{(1)}_{\lambda} (s,x)  = \; &\frac{e^{ix^2/4s}}{(8\pi^2 i s)^{3/2} |x-2\lambda T G(1)|} \int \rd y \, \frac{\Delta^m_y e^{-iy \cdot (x/2s + \lambda F(1))}}{(-1)^m \, |x/2s + \lambda F(1)|^{2m}} (e^{iy^2/4s} - 1) h_\lambda (s,y)
\end{split} \]
we find by integrating by parts that
\[ \begin{split}
| R^{(1)}_{\lambda} (s,x) | \leq \; &\frac{C}{s^{3/2} |x-2\lambda T G(1)| |x/2s +\lambda F(1)|^{2m}} \\ &\times \sum_{|\alpha|+|\beta| = 2m} \int \rd y \, \left| D^{\alpha} (e^{iy^2/4s} - 1) \right| \, |D^{\beta} h_\lambda (s,y)| \,.
\end{split} \]
Using that \begin{equation}\label{eq:Dal-pari} |D^{\alpha} (e^{iy^2/4s} - 1)| \leq \frac{C}{s^r} \, \left(1 + \frac{|y|^{2r}}{s^r}\right) \end{equation} if $|\alpha| = 2r$, $r \geq 1$, and that  
\begin{equation}\label{eq:Dal-dis} |D^{\alpha} (e^{iy^2/4s} - 1)| \leq \frac{C|y|}{s^r} \left(1 + \frac{|y|^{2(r-1)}}{s^{r-1}}\right) \end{equation} if $|\alpha| = 2r-1$, $r \geq 1$, we arrive at 
\[ \begin{split}
| R^{(1)}_{\lambda} (s,x)|  \leq\; &\frac{CR^{-3/2} K_0^{2m}}{s^{3/2} |x-2\lambda T G(1)| |x/2s +\lambda F(1)|^{2m}} \left( \frac{K_0T}{R} \right)^{2n} \left[ 1 + \frac{Z}{\lambda} \log^{2n} (1+s/T)\right] \\ &\times \left\{  \int \rd y \, \frac{|y|^2}{s} \, \frac{1}{1+(|y| / R)^{2n}} +
\sum_{r=1}^m \frac{1}{(K_0^2 s)^r} \int \rd y \, \left( 1+ \frac{y^{2r}}{s^r} \right) \frac{1}{1+(|y|/R)^{2n}} \right. \\ 
&\left. \hspace{1cm}+ 
\sum_{r=1}^m \frac{1}{(K_0^2 s)^r} \int \rd y \,  (K_0 |y|) \, \left( 1+ \frac{y^{2(r-1)}}{s^{r-1}} \right) \frac{1}{1+(|y|/R)^{2n}} \right\},
\end{split} \]
where the first term in the parenthesis corresponds to $|\alpha| = 0$, the second to $|\alpha| = 2r$ and the third to $|\alpha| = 2r-1$. It follows that
\begin{equation}\label{eq:R12} \begin{split}
| R^{(1)}_{\lambda} (s,x)|  \leq &\; \frac{CR^{3/2} K_0^{2m}}{s^{3/2} |x-2\lambda T G(1)| |x/2s +\lambda F(1)|^{2m}} \left( \frac{K_0T}{R} \right)^{2n} \left[ 1 + \frac{Z}{\lambda}  \log^{2n} (1+s/T)\right] \\ &\times \left\{ \frac{R^2}{s} + \frac{R^2}{s} \sum_{r=1}^m \frac{1}{(RK_0)^{2r-1}} \left(1 + \left(\frac{R^2}{s}\right)^{2r} \right)  \right\} \\
\leq &\; \frac{CR^{7/2} K_0^{2m}}{s^{5/2} |x-2\lambda T G(1)| |x/2s +\lambda F(1)|^{2m}} \left( \frac{K_0 T}{R} \right)^{2n} \left[ 1 + \frac{Z}{\lambda} \log^{2n} (1+s/T)\right]   \\ &\times \left(1 + \left(\frac{R^2}{s}\right)^{2m} \right) 
\end{split}\end{equation}
for all $n > m +3/2$, and all $m \geq 1$. 
Combining this bound with (\ref{eq:R11}), we find that
\[ \begin{split}
| R^{(1)}_{\lambda} (s,x)|  \leq &\; \frac{CR^{7/2}}{s^{5/2} |x-2\lambda T G(1)| (1 + (|x/2s +\lambda F(1)| / K_0)^{2m})} \left( \frac{K_0 T}{R} \right)^{2n}   \\ &\times  \left[ 1 + \frac{Z}{\lambda} \log^{2n} (1+s/T)\right]  \left(1 + \left(\frac{R^2}{s}\right)^{2m} \right) 
\end{split}\]
for all $n > m +3/2$, and all $m \geq 1$. 
We thus conclude that
\begin{equation}
\begin{split}
\| R^{(1)}_\lambda (s,x) \| \leq \; &\frac{CR^{7/2}}{s^{5/2}} \left( \frac{K_0 T}{R} \right)^{2n} 
\left[ 1 + \frac{Z}{\lambda} \log^{2n} (1+s/T)\right]  \left(1 + \left(\frac{R^2}{s}\right)^{2m} \right) \\ &\times \left( \int \frac{\rd x}{|x-2\lambda T G(1)|^2 (1 + (|x/2s +\lambda F(1)| / K_0)^{2m})^2} \right)^{1/2}  \,.
\end{split}\end{equation}
Since \[ \begin{split} 
 \int \frac{\rd x}{|x-2\lambda T G(1)|^2} & \, \frac{1}{(1 + (|x/2s +\lambda F(1)| / K_0)^{2m})^2} \\ &= 2sK_0 \int \frac{\rd x}{|x- \frac{\lambda T G(1+s/T)}{sK_0}|^2 \, (1+x^{2m})^2} \leq \frac{C (sK_0)^3}{\lambda^2 T^2 |G(1+s/T)|^2} \end{split} \]
we find that
\begin{equation}
\begin{split}
\| R^{(1)}_\lambda (s,x) \|
\leq \; &\frac{CR^{7/2} K_0^{3/2}}{\lambda s T |G(1+s/T)|} \left( \frac{K_0 T}{R} \right)^{2n} 
\left[ 1 + \frac{Z}{\lambda} \log^{2n} (1+s/T)\right]  \left(1 + \left(\frac{R^2}{s}\right)^{2m} \right) 
\end{split}\end{equation}
for any $m \geq 1$. Since $|G(1+s/2T)| \geq C (1+s/2T)$, we find
\begin{equation}\label{eq:R13} \begin{split} 
\int_{t_0}^{t-T} \rd s\, & \| R^{(1)}_\lambda (s,x )\| \\   \leq 
\, & 
\frac{CR^{7/2} K_0^{3/2}}{\lambda T} \left( \frac{K_0 T}{R} \right)^{2n}  \int_{t_0/T}^{\infty} \frac{\rd s}{s (1+s)}
\left[ 1 + \frac{Z}{\lambda} \log^{2n} (1+s)\right]  \left(1 + \left(\frac{R^2}{T}\right)^{2m} \frac{1}{s^{2m}} \right) 
\\ \leq 
\, & 
\frac{C(K_0 R)^{3/2}}{\lambda} \left( \frac{K_0 T}{R} \right)^{2n}  \frac{R^2}{t_0} \left(1+ \left(\frac{R^2}{t_0} \right)^{2m} \right) \left( 1 + \frac{Z}{\lambda} \left(\frac{t_0}{T}\right)^{\eps} \right) 
\end{split}\end{equation}
for any $m \geq 1$ and $n >m+3/2$, and any $\eps >0$.

\medskip

To control the second remainder term on the r.h.s. of (\ref{eq:RR}), we  write 
\[ \begin{split} \Big( \frac{1}{|2sp - 2\lambda TG(1)|} \, &e^{i Z  \int_0^s \frac{\rd \tau}{|2\tau p - 2 \lambda T G(1)|}}  e^{-i x \cdot A(T)} \psi_T \Big)(y) \\ =& \, \frac{1}{(2\pi)^{3/2}} \int \frac{\rd k}{|2sk - 2\lambda T G(1)|}  \,  e^{ik\cdot y} e^{i Z   \int_0^s \frac{\rd \tau}{|2\tau k - 2 \lambda T G(1)|}} \wh{\psi_T} (k + \lambda F(1)) \\ = &\; \frac{e^{-i\lambda F(1) \cdot y}}{(2\pi)^{3/2}} \int \frac{\rd k}{|2sk - 2\lambda T G(1+s/T)|} \, e^{ik\cdot y} e^{i Z  \int_0^s \frac{\rd \tau}{|2\tau k - 2 \lambda T G(1+\tau/T)|}} \wh{\psi_T} (k) \,. \end{split}\] 
Hence
\[ 
R^{(2)}_\lambda (s,x) = \frac{e^{ix^2/4s}}{(8\pi^2 i s)^{3/2}} \int \rd y \, e^{-i y \cdot (x/2s + \lambda F(1))} \, (e^{iy^2/4s} - 1) \, g_\lambda (s,y) \]
with
\[ g_\lambda (s,y) =  \int \frac{\rd k}{|2sk - 2\lambda T G(1)|} \, e^{ik\cdot y} e^{i Z  \int_0^s \frac{\rd \tau}{|2\tau k - 2 \lambda T G(1+\tau/T)|}} \wh{\psi_T} (k) \,. \] 
In Lemma \ref{lm:g}, we show that for every multi-index $\beta \in \bN^3$,
\[ \left| D^{\beta}_x g_\lambda (s,x) \right| 
\leq \frac{C}{\lambda (T+s)}
\frac{R^{-3/2}  K_0^{|\beta|}}{1+(x/R)^{2n}} \,\left(\frac{K_0 T}{R}\right)^{2n} \left[ 1 + \frac{Z}{\lambda} \log^{2n} (1+ s/T) \right] \,. \] 

Therefore, on the one hand
\begin{equation}\label{eq:R21} \begin{split} |R^{(2)}_\lambda (s,x)| \leq \;& \frac{C}{s^{5/2}} \int \rd y  \, |y|^2 \, |g_\lambda (s,y)| \\ \leq \; &
\frac{C R^{7/2}}{\lambda \, s^{5/2}\, (s+T)} \left( \frac{TK_0}{R} \right)^{2n} \left[ 1 + \frac{Z}{\lambda} \log^{2n} (1+s/T) \right] \end{split} \end{equation}
for all $n > 3/2$. On the other hand, from 
\[ \begin{split} 
R^{(2)}_{\lambda} (s,x)  = \; &\frac{e^{ix^2/4s}}{(8\pi2 is)^{3/2}} \int \rd y \, \frac{\Delta^m_y e^{-iy \cdot (x/2s + \lambda F(1))}}{(-1)^m \, |x/2s + \lambda F(1)|^{2m}} (e^{iy^2/4s} - 1) g_\lambda (s,y)
\end{split} \]
we find by integrating by parts that
\[ \begin{split}
| R^{(2)}_{\lambda} (s,x) | \leq \; &\frac{C}{s^{3/2} |x/2s +\lambda F(1)|^{2m}} \sum_{|\alpha|+|\beta| = 2m} \int \rd y \, \left| D^{\alpha} (e^{iy^2/4s} - 1) \right| \, |D^{\beta} g_\lambda (s,y)| \,.
\end{split} \]
Using the bounds (\ref{eq:Dal-pari}), (\ref{eq:Dal-dis}), we obtain, similarly to (\ref{eq:R12}), the bound
\[ \begin{split}
| R^{(2)}_{\lambda} (s,x)| 
\leq &\; \frac{CR^{7/2} K_0^{2m}}{\lambda  \, s^{5/2} \, (s+T) |x/2s +\lambda F(1)|^{2m}} \left( \frac{K_0 T}{R} \right)^{2n} \left[ 1 + \frac{Z}{\lambda} \log^{2n} (1+s/T)\right]   \\ &\times \left(1 + \left(\frac{R^2}{s}\right)^{2m} \right) 
\end{split}\]
for all $n > m +3/2$, and all $m \geq 1$. Combining the last bound with (\ref{eq:R21}), we conclude that
\[ \begin{split}
| R^{(2)}_{\lambda} (s,x)| 
\leq &\; \frac{CR^{7/2}}{\lambda  \, s^{5/2} \, (s+T) \, (1 + (|x/2s +\lambda F(1)|/ K_0)^{2m})} \left( \frac{K_0 T}{R} \right)^{2n} \left[ 1 + \frac{Z}{\lambda} \log^{2n} (1+s/T)\right]   \\ &\times \left(1 + \left(\frac{R^2}{s}\right)^{2m} \right) \,.
\end{split}\]
Hence we have
\begin{equation}
\begin{split}
\| R^{(2)}_\lambda (s,x) \| \leq \; &\frac{CR^{7/2}}{\lambda \, s^{5/2} (s+T)} \left( \frac{K_0 T}{R} \right)^{2n} 
\left[ 1 + \frac{Z}{\lambda} \log^{2n} (1+s/T)\right]  \left(1 + \left(\frac{R^2}{s}\right)^{2m} \right) \\ &\times \left( \int \frac{\rd x}{(1 + (|x/2s +\lambda F(1)| / K_0)^{2m})^2} \right)^{1/2}  \\
\leq \; &\frac{CR^{7/2} K_0^{3/2}}{\lambda s (T+s)} \left( \frac{K_0 T}{R} \right)^{2n} 
\left[ 1 + \frac{Z}{\lambda} \log^{2n} (1+s/T)\right]  \left(1 + \left(\frac{R^2}{s}\right)^{2m} \right) 
\end{split}\end{equation}
for all $m \geq 1$. Similarly to (\ref{eq:R13}), this implies that  
\begin{equation}\label{eq:R23} \begin{split} 
\int_{t_0}^{t-T} \rd s\,  \| R^{(2)}_\lambda (s,x )\| \leq 
\, & 
\frac{C(K_0 R)^{3/2}}{\lambda} \left( \frac{K_0 T}{R} \right)^{2n}  \frac{R^2}{t_0} \left(1+ \left(\frac{R^2}{t_0} \right)^{2m} \right) \left( 1 + \frac{Z}{\lambda} \left(\frac{t_0}{T}\right)^{\eps} \right) ,
\end{split}\end{equation}
for any $m \geq 1$, $n >m+3/2$, and $\eps >0$.

\medskip

{F}rom (\ref{eq:cou-int}), (\ref{eq:cou-int2}), (\ref{eq:R13}), (\ref{eq:R23}), we find that
\[ \begin{split} 
\Big\|  \Big[ &e^{-i(t-T) [p^2 - Z /|x-2\lambda TG(1)|]} - e^{-i(t-T) p^2} e^{i Z  \int_0^{t-T} \frac{\rd \tau}{|\tau p -2\lambda T G(1)|} } \Big] e^{-i x \cdot A(T)} \psi_T \Big\| \\
\leq \; & \frac{C Z t_0}{R} + C \frac{Z}{\lambda}
 \, (K_0 R)^{3/2} \left( \frac{K_0 T}{R} \right)^{2n}  \frac{R^2}{t_0} \left(1+ \left(\frac{R^2}{t_0} \right)^{2m} \right) \left( 1 + \frac{Z}{\lambda} \left(\frac{t_0}{T}\right)^{\eps} \right) 
\end{split}\] 
for all $m \geq 1$, $n >m+3/2$, and $\eps >0$. We now choose $m=1$ and $n=3$, and we set \[ \frac{t_0}{R^2} = \left( \frac{(K_0R)^{3/2} \left( \frac{K_0 T}{R} \right)^6}{R\lambda} \right)^{1/4} \,. \] 
We will choose $K_0$ so that $K_0 R$ and $K_0 T/R$ are large in the limit of large $(R\lambda)$ so that we may assume that $t_0/R^2 \leq 1$, $(t_0/T)^{\eps} \leq \lambda / (Z)$.  Then
\[ \begin{split} 
\Big\|  \Big[ e^{-i(t-T) [p^2 - Z /|x-2\lambda T G(1)|]} - e^{-i(t-T) p^2} e^{i Z  \int_0^{t-T} \frac{\rd \tau}{|\tau p -2\lambda T G(1)|} } \Big] &e^{-i x \cdot A(T)} \psi_T \Big\| \\
\leq \; & C Z R \left( \frac{(K_0R)^{3/2} \left( \frac{K_0 T}{R} \right)^6}{R\lambda} \right)^{1/4} \,. \end{split}\]

This last bound, together with (\ref{eq:fr}), implies that
\[ \begin{split} 
\| \chi_{\delta,\theta} (t) \, \cU(t,T) e^{-i\int_0^T \rd \tau (p-A(\tau))^2} \psi \| \geq \; &
1 - \frac{C}{R \lambda} \left(1+\frac{R^2}{t} \right) - C (RK_0)^{-5/2} \\ &- C Z R \left( \frac{(K_0R)^{3/2} \left( \frac{K_0 T}{R} \right)^6}{R\lambda} \right)^{1/4} \,. \end{split}\]
We finally optimize our choice of $K_0$. To this end, we set $(K_0 T/R) = C_0 (R \lambda)^{2/35}$, for an appropriate constant $C_0$. This yields 
\[ \begin{split} 
\| \chi_{\delta,\theta} (t) \, \cU(t,T) e^{-i\int_0^T \rd \tau (p-A(\tau))^2} \psi \| \geq \; &
1 - \frac{C}{R \lambda} \left(1+\frac{R^2}{t} \right) - \frac{C}{(R\lambda)^{1/7}}  \left( \frac{Z T^{3/2}}{R^2} \right)^{4/7} \, , \end{split}\]
which, combined with (\ref{eq:cou0}), concludes the proof of the theorem.
\end{proof}

\begin{lemma}\label{lm:free}
Suppose that $\chi \in C_0^{\infty} (\bR^3)$, with $\chi (x) = 0$ for all $|x| \geq 1$ and $\chi (x) = 1$ for all $|x| \leq 1/2$. Assume that $C^{-1} \leq R^2 /T \leq C$, $Z \leq \lambda$, $K_0 \leq C \lambda$ for an appropriate constant $C$, and that 
$\lambda R \geq 1$ is large enough (at the end $(K_0R) \simeq (\lambda R)^{2/35}$, and therefore the condition $K_0 \leq C \lambda$ is satisfied for sufficiently large $(R\lambda)$). Then, for every $t \geq T$, and for every constant $D>0$, we have that  
\begin{equation}\label{eq:lmfree} \begin{split} 
\| {\bf 1} (|x| \geq D t) \, e^{-i t p^2} e^{i Z  \int_0^{t-T} \frac{\rd \tau}{|2\tau p -2\lambda T G(1+\tau/T)|} } \chi (p / K_0) \psi \| \leq \frac{C}{D R} \left(1+\frac{R^2}{t}\right) \,.\end{split}\end{equation}
\end{lemma}

\begin{proof}
We notice that 
\begin{equation}\label{eq:free0}
\left\| {\bf 1} (|x| \geq D t) \, e^{-i t p^2} e^{i Z  \int_0^{t-T} \frac{\rd \tau}{|2\tau p -2\lambda T G(1+\tau/T)|} } \chi (p / K_0) \psi \right\|^2 \leq \frac{1}{(D t)^2} W^2(t) \end{equation}
where \[\begin{split}  W^2 (t) := \langle  e^{-i t p^2} &e^{i Z  \int_0^{t-T} \frac{\rd \tau}{|2\tau p -2\lambda T G(1+\tau/T)|} } \chi (p / K_0) \psi, \\ & x^2 \, e^{-i t p^2} e^{i Z  \int_0^{t-T} \frac{\rd \tau}{|2\tau p -2\lambda T G(1+\tau/T)|} } \chi (p / K_0) \psi \rangle \,.\end{split}\]
Next we compute
\[ \begin{split}
\frac{\rd}{\rd t} W^2 (t) = \; &
\langle e^{-i t p^2} e^{i Z  \int_0^{t-T} \frac{\rd \tau}{|2\tau p -2\lambda T G(1+\tau/T)|} } \chi (p / K_0) \psi, \\ &\hspace{3cm} i \, \left[ p^2 +  \frac{Z}{|2(t-T) p -2 \lambda T G(t/T)|},  x^2 \right] \\ &\hspace{5cm} \times e^{-i t p^2} e^{i Z  \int_0^{t-T} \frac{\rd \tau}{|2\tau p -2\lambda T G(1+\tau/T)|} } \chi (p / K_0) \psi \rangle 
\\  = \; &2 \text{Im} \, 
\Big\langle e^{-i t p^2} e^{i Z  \int_0^{t-T} \frac{\rd \tau}{|2\tau p -2\lambda T G(1+\tau/T)|} } \chi (p / K_0) \psi,  \\ &\hspace{3cm}  x \cdot \left( 2 p + 2 Z (t-T) \frac{2(t-T) p - 2 \lambda T G(t/T)}{|2(t-T)p - 2\lambda T G(t/T)|^3} \right)  \\ &\hspace{5cm} \times  e^{-i t p^2} e^{i Z  \int_0^{t-T} \frac{\rd \tau}{|\tau p -2\lambda T G(1+\tau/2T)|} } \chi (p / K_0) \psi \Big\rangle
\end{split}\]
which, using $|2(t-T) p - 2\lambda T G(t/T)| \geq (C \lambda - K_0) t$, implies that
\begin{equation}\label{eq:free1} \begin{split}
\Big| \frac{\rd}{\rd t} W^2(t) \Big| \leq \; & C W(t) \, \left( \frac{Z}{t \lambda^2} + \| |p|   \psi \| \right) \\ \leq \; & C R^{-1} W(t)  \end{split} \end{equation}
for all $t >T$ (because $\lambda R \geq 1$, and $Z/\lambda \leq 1$, $\lambda T/R \geq 1$). By Gronwall's Lemma we find that
\begin{equation}\label{eq:free2} W(t) \leq C (t-T) R^{-1} + W(T) \end{equation}
where 
\[ W^2 (T) = \langle e^{-i T p^2} \chi (p / K_0) \psi, x^2 \, e^{-i T p^2}  \chi (p / K_0) \psi \rangle \,. \]
Similarly to (\ref{eq:free1}), we find that \[ 
\left| \frac{\rd}{\rd T} W^2(T) \right| \leq 2 W(T) \, \| |p| \psi \| \leq CR^{-1} W(T) \]
which implies that \[ W(T) \leq CTR^{-1} + W(0) \leq C \left( T R^{-1} + R \right) \]
and thus, combining the last equation with (\ref{eq:free0}) and (\ref{eq:free2}), we obtain (\ref{eq:lmfree}).
\end{proof}

\bigskip

\begin{lemma}\label{lm:h}
Let 
\[ h_\lambda (s,x) = \int \rd k \, e^{ik\cdot y} e^{i Z  \int_0^s \frac{\rd \tau}{|2\tau k - 2 \lambda T G(1+\tau/T)|}} \wh{\psi_T} (k) \] 
with $\wh{\psi}_T (k) = e^{-iTk^2} \chi (k/K_0) \psi$, with $\chi \in C^{\infty}_0 (\bR^3)$ such that $\chi (y)= 1$ for $|y| \leq 1/2$ and $\chi (y) = 0$ for $|y| \geq 1$. Assume that $R^{-1} +R T^{-1} \leq K_0 \leq C \lambda$ for an appropriate constant $C$ (at the end, we will choose $K_0 R \simeq (R \lambda)^{2/35}$, and therefore these conditions are satisfied for large enough $\lambda R$).  Assume also $Z \leq \lambda$. Then, for every $\beta \in \bN^3$, we have 
\[ \begin{split} \left| D^{\beta}_x h_\lambda (s,x) \right| \leq  \, &\frac{ C R^{-3/2}  K_0^{|\beta|}}{1+(x/R)^{2n}} \,\left(\frac{K_0 T}{R}\right)^{2n} \left[ 1 + \frac{Z}{\lambda} \log^{2n} (1+ s/T) \right] \,. \end{split} \]
\end{lemma}

\begin{proof}
We have
\begin{equation}\label{eq:hlam1}  D^{\beta}_x h_\lambda (s,x) = \int \rd k \, (ik)^{\beta} \, e^{ik\cdot x} e^{iZ \int_0^s \frac{\rd
\tau}{|2\tau k - 2 \lambda T G(1+\tau/T)|}} \wh{\psi}_T (k) \,. \end{equation}
Hence
\begin{equation}\label{eq:h1} \begin{split} | D_x^\beta h_\lambda (s,x)| \leq \; & \int \rd k \, |k|^{\beta}  \chi (k / K_0) |\wh{\psi} (k)| \\ \leq \; &
R^{-3/2-|\beta|} \int \rd k \, |k|^{\beta}  \chi (k / R K_0) \frac{1}{(1+ k^2)^2} \\ \leq \; &C R^{-3/2-|\beta|} \int_{|k| \leq 2 R K_0} \frac{1}{(1+|k|)^{4-|\beta|}} \\ \leq \; &C R^{-3/2} K_0^{|\beta|} \,. \end{split} \end{equation}

Integrating by parts in (\ref{eq:hlam1}), we arrive at
\[ \begin{split}
 D^{\beta}_x h_\lambda (s,x) = & \; \int \rd k \, (ik)^{\beta} \, \frac{\Delta_k^n \, e^{ik\cdot x}}{(-1)^n |x|^{2n}} \, e^{iZ \int_0^s \frac{\rd
\tau}{|2\tau k - 2\lambda T G(1+\tau/T)|}} \wh{\psi}_T (k) \\  = & \; 
\int \rd k \, \frac{e^{ik\cdot x}}{(-1)^n |x|^{2n}} \, \Delta^n \left[ (ik)^{\beta} \, e^{iZ \int_0^s \frac{\rd \tau}{|2\tau k - 2\lambda T G(1+\tau/T)|}} \wh{\psi}_T (k) \right]
\end{split} \]
and therefore
\begin{equation}\label{eq:dbh}  \begin{split}
| D^{\beta}_x h_\lambda (s,x)| \leq & \; \frac{C}{|x|^{2n}} \, \sum_{|\alpha_1| + | \alpha_2|+|\alpha_3| = 2n} 
\int \rd k \,|k|^{|\beta|-|\alpha_1|} \, |D^{\alpha_2} \, e^{iZ \int_0^s \frac{\rd \tau}{|2\tau k - 2\lambda T G(1+\tau/T)|}}| \, |D^{\alpha_3} \, \wh{\psi}_T (k)| \,. 
\end{split}\end{equation}
Observe that, for all $|\alpha_2| \geq 1$, 
\[  |D^{\alpha_2} \, e^{iZ \int_0^s \frac{\rd \tau}{|2\tau k - 2\lambda T G(1+\tau/T)|}}| \leq C \sum_{m=1}^{|\alpha_2|} \sum_{j_1,..,j_m \geq 1: j_1+..+j_m=|\alpha_2|} \prod_{i=1}^m Z \int_0^s \frac{\rd\tau \, \tau^{j_i}}{|2\tau k - 2\lambda T G(1+\tau/T)|^{j_i+1}}\,. \]
Using the fact that $|k| \leq K_0$ on the support of $\wh{\psi}_T$, we find $|2\tau k - 2\lambda T G(1+\tau/T)| \geq 2 \lambda T |G(1+\tau/T)| -2 K_0 \tau \geq C\lambda T + (C\lambda- K_0) \tau$. Therefore, assuming that $K_0 < C \lambda /2$, and that $Z <  \lambda$, 
\begin{equation}\label{eq:a2} \begin{split} |D^{\alpha_2} \, e^{iZ \int_0^s \frac{\rd \tau}{|2\tau k - 2\lambda T G(1+\tau/T)|}}| \leq \; & C \sum_{m=1}^{|\alpha_2|} \sum_{j_1,..,j_m \geq 1: j_1+..+j_m=|\alpha_2|} \prod_{i=1}^m \frac{Z}{\lambda^{j_1+1}} \int_0^s \frac{\rd\tau}{T+\tau} \\ \leq \;&  \frac{C}{\lambda^{|\alpha_2|}} \sum_{m=1}^{|\alpha_2|} \left(\frac{Z \log (1+s/t)}{\lambda}\right)^{m} \\ \leq \; &
C \frac{Z}{\lambda^{|\alpha_2|+1}} (1+\log^{|\alpha_2|} (1+s/T)) \,. \end{split} \end{equation}

On the other hand, by a simple computation, we have
\begin{equation}\label{eq:a3} \begin{split}  | D^{\alpha_3} \wh{\psi}_T (k)| \leq \; &C \chi (k /K_0) \, \frac{R^{3/2+|\alpha_3|}}{(1+(Rk)^2)^2} \left( \frac{T^{1/2}}{R} (1+|k| T^{1/2}) + \frac{1}{RK_0} +\frac{1}{\sqrt{1+(Rk)^2}} \right)^{|\alpha_3|} \\ \leq \; &C\chi (k/K_0) \, \frac{R^{3/2}}{(1+(Rk)^2)^2} (K_0 T)^{|\alpha_3|} \end{split} \end{equation}
assuming that $K_0 R \geq 1$ and $K_0 T/R \geq 1$.

{F}rom (\ref{eq:dbh}), it follows that
\[ \begin{split} 
|x|^{2n} \, | D^{\beta}_x h_\lambda (s,x)| \leq \; &C \sum_{|\alpha_1|+|\alpha_3|=2n} \int_{|k| \leq K_0}  \rd k \, |k|^{|\beta|-|\alpha_1|} (K_0 T)^{|\alpha_3|} \frac{R^{3/2}}{(1+(Rk)^2)^2} 
\\ &+ C\sum_{|\alpha_1|+|\alpha_2|+|\alpha_3| = 2n, |\alpha_2| \geq 1} \frac{Z}{\lambda^{|\alpha_2|+1}} (1+ \log^{|\alpha_2|} (1+s/T)) \\ &\hspace{1cm} \times \int_{|k| \leq K_0}  \rd k \, |k|^{|\beta|-|\alpha_1|} (K_0 T)^{|\alpha_3|} \frac{R^{3/2}}{(1+(Rk)^2)^2} 
\\ \leq \; & CR^{-3/2-|\beta| +|\alpha_1|} \sum_{|\alpha_1|+|\alpha_3|=2n} (K_0 T)^{|\alpha_3|} \int_{|k| \leq R K_0}   \frac{\rd k}{(1+|k|)^{4-|\beta| +|\alpha_1|}} 
\\ &+ C\frac{ZR^{-3/2 -|\beta| +|\alpha_1|} }{\lambda} (1+\log^{2n} (1+s/T)) 
\\ &
\hspace{1cm}\times \sum_{|\alpha_1|+|\alpha_2|+|\alpha_3| = 2n, |\alpha_2| \geq 1} \frac{1}{\lambda^{|\alpha_2|}}  (K_0 T)^{|\alpha_3|}
\int_{|k| \leq R K_0}   \frac{\rd k}{(1+|k|)^{4-|\beta| +|\alpha_1|}} 
\end{split}\]
and therefore
\[ \begin{split}
|x|^{2n} \, | D^{\beta}_x h_\lambda (s,x)| \leq \; & CR^{-3/2+2n -|\beta| } \sum_{|\alpha_1|+|\alpha_3|=2n} \left(\frac{K_0 T}{R} \right)^{|\alpha_3|} (1+ (K_0 R)^{-1+|\beta| -|\alpha_1|+\eps})  \\ &+C \frac{Z R^{-3/2+ 2n-|\beta|}}{\lambda} (1+\log^{2n} (1+s/T)) 
\\ &
\hspace{1cm}\times
\sum_{|\alpha_1|+|\alpha_2|+|\alpha_3| = 2n, |\alpha_2| \geq 1} \frac{1}{(R\lambda)^{|\alpha_2|}}  \left( \frac{K_0 T}{R } \right)^{|\alpha_3|}(1+ (K_0 R)^{-1+|\beta| -|\alpha_1|+\eps}) 
 \\ \leq \; &C R^{-3/2+2n} K_0^{|\beta|}  \left(\frac{K_0 T}{R}\right)^{2n} \left[ 1 + \frac{Z}{\lambda} \log^{2n} (1+ s/T) \right] \,.
 \end{split} \]
Combining this bound with (\ref{eq:h1}), we find, for arbitrary $n\geq 0$,
\[ |D^{\beta}_x h_\lambda (s,x)| \leq  \, \frac{ C R^{-3/2}  K_0^{|\beta|}}{1+(x/R)^{2n}} \,\left(\frac{K_0 T}{R}\right)^{2n} \left[ 1 + \frac{Z}{\lambda} \log^{2n} (1+ s/T) \right] \,. \]
\end{proof} 

\begin{lemma}\label{lm:g}
Let 
\[ g_\lambda (s,x) = \int \rd k \, e^{i k\cdot x}\frac{e^{i Z \int_0^s \frac{\rd \tau}{|2\tau k - 2\lambda T G(1+\tau/T)|}}}{|2sk - 2\lambda T G(1+s/T)|} \,  \wh{\psi}_T (k) \]
with $\wh{\psi}_T (k) = e^{-iTk^2} \chi (k/K_0) \psi$, with $\chi \in C^{\infty}_0 (\bR^3)$ such that $\chi (y)= 1$ for $|y| \leq 1/2$ and $\chi (y) = 0$ for $|y| \geq 1$. Assume that $R^{-1} +R T^{-1} \leq K_0 \leq C \lambda$ for an appropriate constant $C$ (at the end, we will choose $K_0 R \simeq (R \lambda)^{2/35}$, and therefore these conditions are satisfied for large enough $\lambda R$).  Assume also $Z \leq \lambda$. Then, for every $\beta \in \bN^3$, we have 
\[\left| D^{\beta}_x g_\lambda (s,x) \right| \leq \frac{C}{\lambda (T+s)}
\frac{R^{-3/2}  K_0^{|\beta|}}{1+(x/R)^{2n}} \,\left(\frac{K_0 T}{R}\right)^{2n} \left[ 1 + \frac{Z}{\lambda} \log^{2n} (1+ s/T) \right] \,. \]   
\end{lemma}

\begin{proof}
We have
\begin{equation}\label{eq:glam1}  D^{\beta}_x g_\lambda (s,x) = \int \rd k \, (ik)^{\beta} \, e^{ik\cdot x} \frac{e^{iZ \int_0^s \frac{\rd
\tau}{|2\tau k -2 \lambda T G(1+\tau/T)|}}}{|2sk - 2\lambda T G(1+s/T)|} \, \wh{\psi}_T (k) \end{equation}
Since, for $|k| \leq K_0$, we have $|sk - 2\lambda T G(1+s/2T)| \geq C \lambda (T+s)- sK_0 \geq C\lambda T + s (C\lambda - K_0) \leq C \lambda (T+s)$, we find 
\begin{equation}\label{eq:g1} \begin{split} | D_x^\beta g_\lambda (s,x)| \leq \; & \frac{C}{\lambda (T+s)} \int \rd k \, |k|^{\beta}  \chi (k / K_0) |\wh{\psi} (k)| \\ \leq \; &
 \frac{C R^{-3/2 - |\beta|}}{\lambda (T+s)}  \int \rd k \, |k|^{\beta}  \chi (k / R K_0) \frac{1}{(1+ k^2)^2} \\ \leq \; &\frac{C R^{-3/2} K_0^{|\beta|}}{\lambda (T+s)}\, . \end{split} \end{equation}
Integrating by parts in (\ref{eq:glam1}), we arrive at
\[ \begin{split}
 D^{\beta}_x g_\lambda (s,x) = & \; \int \rd k \, (ik)^{\beta} \, \frac{\Delta_k^n \, e^{ik\cdot x}}{(-1)^n |x|^{2n}} \, \frac{e^{iZ \int_0^s \frac{\rd
\tau}{|2\tau k - 2\lambda T G(1+\tau/T)|}}}{|2sk -2 \lambda T G(1+s/T)|} \,  \wh{\psi}_T (k) \\  = & \; 
\int \rd k \, \frac{e^{ik\cdot x}}{(-1)^n |x|^{2n}} \, \Delta^n \left[ (ik)^{\beta} \, \frac{e^{iZ \int_0^s \frac{\rd \tau}{|2\tau k -2 \lambda T G(1+\tau/T)|}}}{|2sk-2\lambda T G(1+s/T)|} \wh{\psi}_T (k) \right]
\end{split} \]
and therefore
\begin{equation}  \begin{split}
| D^{\beta}_x h_\lambda (s,x)| \leq & \; \frac{C}{|x|^{2n}} \, \sum_{|\alpha_1| + \dots +|\alpha_4| = 2n} 
\int \rd k \,|k|^{|\beta|-|\alpha_1|} \, | D^{\alpha_2} \, e^{iZ \int_0^s \frac{\rd \tau}{|2\tau k - 2\lambda T G(1+\tau/T)|}}| \\ &\hspace{1cm} \times  \frac{s^{|\alpha_3|}}{|2sk - 2\lambda T G(1+s/T)|^{|\alpha_3|+1}} \, |D^{\alpha_3} \, \wh{\psi}_T (k)| \,.
\end{split}\end{equation}
{F}rom (\ref{eq:a2}), (\ref{eq:a3}), and since $|sk-\lambda T G(1+s/T)| \geq C \lambda (T+s)$, we find
\[ \begin{split} | D^{\beta}_x h_\lambda (s,x)| \leq & \;  \frac{C R^{3/2}}{|x|^{2n}} \, \sum_{|\alpha_1| + \dots +|\alpha_4| = 2n, |\alpha_2|=0} \frac{(K_0T)^{|\alpha_4|}}{\lambda^{|\alpha_3|+1} (T+s)} \int_{|k| \leq K_0}  \rd k \,\frac{|k|^{|\beta|-|\alpha_1|}}{(1+(Rk)^2)^2}  \\ &+ \frac{C R^{3/2}}{|x|^{2n}} \frac{Z}{\lambda} (1+\log^{2n} (1+s/T))  \, \sum_{|\alpha_1| + \dots +|\alpha_4| = 2n, |\alpha_2| \geq 1} \frac{(K_0T)^{|\alpha_4|}}{\lambda^{|\alpha_2|+|\alpha_3|+1} (T+s)} \\ &\hspace{1cm} \times \int_{|k| \leq K_0} \rd k \,\frac{|k|^{|\beta|-|\alpha_1|}}{(1+(Rk)^2)^2} \\
\leq & \; \frac{C R^{-3/2}  K_0^{|\beta|}}{\lambda (T+s) (|x|/R)^{2n}}  \left[ 1 + \frac{Z}{\lambda} \log^{2n} (1+s/T) \right] \, \left(\frac{K_0 T}{R} \right)^{2n}
\end{split}\]
where we used $K_0 T/ R \geq 1$, $R\lambda > 1$, $Z < \lambda$. Combining the last equation with (\ref{eq:g1}), we conclude the proof of the lemma.
\end{proof}

\end{document}